\documentclass[final,letterpaper]{siamltex}

\usepackage[T1]{fontenc}

\usepackage{graphicx}
\usepackage{amsmath,amssymb,dsfont}
\usepackage{color}
\usepackage[english]{babel}
\usepackage{setspace,mathrsfs}
\usepackage[mathscr]{euscript}

\newtheorem{assumption}{Assumption}[section]
\newtheorem{remark}{Remark}[section]
\newtheorem{example}{Example}

\let\ALP \mathcal
\let\FLD \mathfrak
\let\VEC \mathbf

\newcommand\transpose{\mathsf{T}}
\newcommand\IND{\mathds{1}}
\newcommand{\ind}[1]{\IND_{\{#1\}}}

\newcommand{\beq}[1]{\begin{eqnarray} #1 \end{eqnarray}}
\newcommand{\beqq}[1]{\begin{eqnarray*} #1 \end{eqnarray*}}
\renewcommand{\Re}{\mathbb{R}}
\newcommand{\Na}{\mathbb{N}}

\newcommand{\ex}[1]{\mathbb{E}\left[#1\right]}
\newcommand{\pr}[1]{\mathbb{P}\left\{#1\right\}}

\newcommand{\ws}{\overset{w^*}{\rightharpoonup}}

\newcommand{\infnorm}[1]{\|#1\|_{\infty}}

\newcommand{\lf}[1]{\lim_{#1\rightarrow \infty}}
\newcommand{\lif}[1]{\underset{#1\rightarrow \infty}{\lim\inf}\:}

\newcommand{\ca}{\textsf{ca}}

\newcommand{\ball}[1]{\mathtt B(#1)}
\let \SF \mathscr
\newcommand{\pj}{\textsf{Pr}}
\renewcommand{\transpose}{\texttt{T}}
\newcommand{\ic}[1]{\textsc{IC}(#1)}

\usepackage[colorlinks]{hyperref}

\title{On the Existence of Optimal Policies for a Class of Static and Sequential Dynamic Teams}
\author{Abhishek Gupta \and Serdar Y\"{u}ksel\thanks{Serdar Y\"{u}ksel is with Department of Mathematics and Statistics, Queen's University, Canada. Email: \texttt{yuksel@mast.queensu.ca}.} \and Tamer Ba\c{s}ar \and C\'{e}dric Langbort
\thanks{Abhishek Gupta, Tamer Ba\c{s}ar and C\'edric Langbort are with Coordinated Science Laboratory at the University of Illinois at Urbana-Champaign (UIUC), USA. Email: {\tt \{gupta54,basar1,langbort\}@illinois.edu}. Research was supported in part by AFOSR MURI Grant FA9550-10-1-0573 and in part by NSF Grant CCF 11-11342. Sections \ref{sec:staticteam} and \ref{sec:unbounded} are generalizations of the results reported in \cite{gupta2014teamacc} at 2014 American Control Conference (ACC). Sections \ref{sec:staticreduction} and \ref{sec:lqgteam} are generalizations of the setting considered in \cite{gupta2014teamcdc}, which has been submitted recently to 53rd IEEE Conference on Decision and Control, 2014.}}
\begin{document}
\maketitle
\begin{abstract}
In this paper, we identify sufficient conditions under which static teams and a class of sequential dynamic teams admit team-optimal solutions. We first investigate the existence of optimal solutions in static teams where the observations of the decision makers are conditionally independent or satisfy certain regularity conditions. Building on these findings and the static reduction method of Witsenhausen, we then extend the analysis to sequential dynamic teams. In particular, we show that a large class of dynamic LQG team problems, including the vector version of the well-known Witsenhausen's counterexample and the Gaussian relay channel problem viewed as a dynamic team, admit team-optimal solutions. Results in this paper substantially broaden the class of stochastic control and team problems with non-classical information known to have optimal solutions.
\end{abstract}

\section{Introduction}
A team is a group of agents who act collectively, but not necessarily sharing all the information, in order to optimize a common cost function. In stochastic teams, there may be probabilistic uncertainty about initial state, observations or cost functions, and/or the evolution of the state is disturbed by some external noise process. The statistics of the noise processes, state evolution equations and observation equations are common knowledge among the agents of a team. At every time step, each agent acquires some information about the state via its observation, may acquire other agents' past and current observations and past actions, and may recall its past observations and actions. If each agent's information is only determined by primitive/exogenous random variables, the team is said to be a static team. If at least one agent's information is affected by an action of another agent, the team is said to be dynamic. The information structure in a team determines such functional and probabilistic relations 
in that team.

A system in which there is a pre-specified order of actions is said to be a sequential team. In non-sequential teams, the ordering of ``who acts when'' is not known {\it a priori}, and ordering of actions may be determined by the outcome of some random process \cite{teneket1996}. In this paper, we focus our attention on the existence of optimal strategies in sequential dynamic teams. A sequential team in which the information available to the decision makers in forward time is non-contracting is said to be a classical team (such settings include the well-studied single-agent Markov Decision Problems with full memory). A sequential team, which is not classical, but has the property that whenever an agent's, say ${\bf A}k$'s, information is affected by the action of some other agent, say ${\bf A}j$, ${\bf A}k$ has access to ${\bf A}j$'s information, is said to have a quasi-classical (or partially nested) information structure. An information structure which is not quasi-classical is said to be non-classical. 
For an extensive discussion on classifications of information structures, the reader is referred to \cite{wit1971b} and \cite{yukselbook}.


We next provide a brief overview of the early developments in the area of teams with asymmetric information, as is relevant to this paper. In 1962, Radner published a seminal paper on team decision problems \cite{radner1962} where he showed that a class of static stochastic teams with continuously differentiable and strictly convex cost functions admit globally optimal (minimizing) solutions which are also the unique person-by-person-optimal solutions; this result was later extended by Krainak et al \cite{krainak1982} to settings with exponential cost functions. A byproduct of Radner's result is that in static linear-quadratic-Gaussian (LQG) team problems, linear strategies of the agents that are person-by-person optimal are also team-optimal. Furthermore, partially nested LQG teams also admit linear optimal solutions, as was observed by Ho and Chu \cite{ho1972a,chu1972}. When the information structure is non-classical, however, Witsenhausen showed that even seemingly simple LQG settings can be very 
difficult to solve \cite{witcount}: He devised a scalar LQG team problem which admits an optimal solution, which, however, is not linear.

Obtaining solutions of classical dynamic team problems is quite well understood, with dynamic programming providing the most convenient approach. For dynamic teams not of the classical type, however, there is no systematic approach which is universally applicable; see \cite{NayyarBookChapter}, \cite{yukselbook}, \cite{CDCTutorial} for a detailed coverage and analysis of the various solution approaches, which also depend on the underlying information structure.

In this paper, we prove that a class of sequential teams with a certain information structure (not-necessarily classical or quasi-classical) admits team-optimal solutions. This constitutes a first step toward understanding the most general conditions under which stochastic dynamic team problems admit optimal solutions. The results obtained and the general framework adopted are applicable to various models of team problems that are studied in economics, information theory, network information theory, and stochastic control.

\subsection{Previous Work}
Some of the earliest and most fundamental works on understanding the role of information in general dynamic stochastic teams were carried out by Witsenhausen in \cite{wit1971b,witcount,wit1971a,wit1988}. Among these fundamental contributions, in \cite{wit1988} Witsenhausen showed that all sequential team problems satisfying an absolute continuity condition of certain conditional measures can be transformed into an equivalent static team problem with a different cost function and with the agents observing mutually independent random variables. This transformation of a dynamic team into a static team problem with independent observations is called \textsl{``static reduction''} of the dynamic team problem (see Section 3.7 and  particularly p. 114 of \cite{yukselbook} both for a discussion on this reduction as well as an overview of Witsenhausen's contributions). In this paper, we make use of this equivalence between dynamic and static team problems to show the existence of optimal strategies in the original 
dynamic team problem. Recently, \cite{bambos2013b} has studied continuous-time stochastic team decision problems that are partly driven by Brownian motion, and has shown that the static reduction of the dynamic team problem can be carried out using the Girsanov transformation.

The existence of optimal strategies for Witsenhausen's counterexample was proved in Witsenhausen's original paper \cite{witcount}, where his proof relied on the structure of the cost function of the team. In that problem, the unique optimal strategy of the second agent involves the conditional mean of the control action of the first agent using an observation that is an additive noise corrupted version of the control action of the first agent, where the noise is a zero-mean unit variance Gaussian random variable. When substituted back, this makes the expected cost functional of the team a non-convex functional on the space of strategies of the first agent. Witsenhausen used several techniques from real analysis to show that an optimal strategy of the first agent exists when the initial state of the system is an arbitrary second-order random variable. A shorter proof of existence for this problem was later presented by Wu and Verd\'u \cite{wu2011} using tools from optimal transport theory. One variant of this 
problem is the Gaussian test channel; there are other variants as well, all with non-classical information \cite{basar2008}. For the Gaussian test channel, proof of existence of optimal strategies (and their derivations) is an indirect one. First, the cost function is lower bounded using the data processing inequality \cite{cover2006}, and then explicit linear strategies are constructed which make the cost achieve the lower bound. Proofs of existence of optimal strategies in some other teams with non-classical information using this method can be found in \cite{bansal1987}.

The main difference between the formulations of Witsenhausen's counterexample and the Gaussian test channel is in the cost functions (even though they are both quadratic) \cite{bansal1987,basar2008}. In the case of Witsenhausen's counterexample, the explicit forms of optimal strategies are not known, whereas in the case of the Gaussian test channel, linear strategies are optimal; more importantly, in both problems (where all random variables are jointly Gaussian), optimal solutions exist. If, however, the distributions of the primitive random variables are discrete with finite support, there is no optimal solution even in the class of behavioral strategies of the agents as illustrated in \cite[p. 90]{yukselbook}. Thus, the question of sufficient conditions for a team problem to admit an optimal solution is an important one, which is addressed in this paper.

Several authors have proven the existence of optimal solutions to stochastic optimization problems through ``lifting''. Specifically, the problem of optimizing an expected cost functional on the space of Borel measurable functions is lifted to an equivalent optimization problem in which the cost functional is defined on the space of probability measures. This technique is heavily used in proving the existence of optimal strategies in Markov decision processes; see for example, \cite{puterman1994, her1996,bert1978}, among several others. This is also the central concept for studying optimal transport problems \cite{amb2008, villani2009}, where the cost function is a measurable function of two random variables with given distributions, and the optimization is performed on the space of joint measures of the random variables given the marginal distributions. In this paper, we use a similar lifting technique to establish the existence of optimal strategies in a class of static and dynamic teams.

A variant of the problem of existence of optimal strategies in stochastic dynamic teams is that of the existence of optimal observation channels in such systems. The relevant question there is how to design observation channels (for example, quantizers) in a team problem so that the overall expected cost is minimized. This problem was studied in \cite{yuksel2012opt}, where some sufficient conditions were obtained on teams and sets of observation channels to ensure the existence of optimal quantizers. There again, the problem of designing quantizers was lifted to one of designing joint probability measures over the state and the observation of an agent satisfying certain constraints. In \cite{yuksel2012opt}, a topology of information channels was introduced based on a fixed input distribution. Related approaches can be found in \cite{borkar1993,borkar2005}.  In this paper, we further generalize these approaches.

\subsection{Outline of the paper}
We prove, under some sufficient conditions, the existence of optimal strategies in a class of sequential team problems. A general stochastic dynamic team problem is formulated in Section \ref{sec:dynamicproblem}. Thereafter, we study two related static team problems in Sections \ref{sec:staticteam} and \ref{sec:unbounded}. In Section \ref{sec:staticteam}, we show that if the cost function of the team is continuous and bounded, and action spaces are compact, then under mild conditions on the observation kernels, a team optimal solution exists. In Section \ref{sec:unbounded}, we extend the result to the case when the action spaces may not be compact and the cost function of the team is a non-negative continuous, possibly unbounded function. In Section \ref{sec:staticreduction}, we show, using the static reduction technique of Witsenhausen \cite{wit1988}, that a large class of dynamic team problems admit team-optimal solutions. We use this result to show, in Section \ref{sec:lqgteam}, that most LQG teams with 
no sharing of observation admit team-optimal solutions. In Section \ref{sec:examples}, we prove the existence of optimal strategies in several LQG team problems of broad interest. Finally, we present concluding remarks in Section \ref{sec:concteam}.

\subsection{Notation}
We introduce here some of the notation used throughout the paper. For a natural number $N$, we let $[N]$ denote the set $\{1,\ldots,N\}$. The set of all non-negative real numbers is denoted by $\Re^+$. If $\ALP X$ is a set and $\ALP A$ is a subset of $\ALP X$, then $\ALP A^\complement$ denotes the complement of the set $\ALP A$.

We use boldfaced letters $\VEC a,\VEC b,\ldots$ to denote generic elements in sets $\ALP A$, $\ALP B$ and so on. If the space $\ALP X$ is the real space, then we simply use $x$ to denote a generic element of $\ALP X$, with no boldface. Uppercase boldfaced letters, for example $\VEC X$, are used to denote random variables. A superscript denotes the index of an agent, while a subscript denotes the time step or an index of a sequence. For example, $\VEC U^i_t$ denotes the control action taken by Agent $i$ at time step $t$. If $\VEC a\in\Re^n$, then we use $\|\VEC a\|_{R}^2$ to denote $\VEC a^\transpose R \VEC a$, for a positive-definite matrix $R$.

Let $X$ be a non-empty set. For a set of elements $\{x_1,\ldots,x_n\}$ in $X$, we let $x_{1:n}$ denote this set. The set $\{x_1,\ldots,x_n\}\setminus\{x_i\}$ is denoted by $x_{-i}$. If $X_1,\ldots,X_n$ are non-empty sets, then $X_{1:n}$ and $X_{-i}$ are shorthand notations for the product sets $X_1\times\ldots\times X_n$ and $X_1\times\ldots\times X_{i-1}\times X_{i+1}\times\ldots\times X_n$, respectively. If we write $x_{1:n}\in X_{1:n}$, then it is construed as $x_1\in X_1$, $x_2\in X_2$, and so on. Similarly, $x^{1:n}$ denotes the set $\{x^1,\ldots,x^n\}$, and $x^{1:n}_{1:t}$ denotes $\{x^1_1,\ldots,x^n_1,\ldots,x^1_t,\ldots,x^n_t\}$, where each element is in the appropriate space.

Let $\ALP X$ be a topological space. The vector space of all bounded continuous functions on $\ALP X$ endowed with the supremum norm $\infnorm{\cdot}$ is denoted by $C_b(\ALP X)$, that is, $C_b(\ALP X):=\{f:\ALP X\rightarrow \Re: f \text{ is continuous and }\infnorm{f}<\infty\}$. The vector space $C(\ALP X)$ is the space of all continuous functions on the topological space $\ALP X$, which are possibly unbounded. Thus, $C_b(\ALP X)\subset C(\ALP X)$. If $\ALP X$ is a metric space, then the vector space $U_b(\ALP X)\subset C_b(\ALP X)$ denotes the set of all bounded uniformly continuous functions on $\ALP X$. The Borel $\sigma$-algebra on $\ALP X$ is denoted by $\FLD B(\ALP X)$. The spaces $\ca(\ALP X)$ and $\wp(\ALP X)$ denote, respectively, the vector space of countably additive signed measures on $\ALP X$ and the set of all probability measures on $\ALP X$ endowed with weak* topology. If $\mu^i\in\wp(\ALP A^i), i\in[N]$ are probability measures, then $d\mu^{1:N}$ denotes the product measure $\mu^1(d\VEC a^1)\cdots\mu^N(d\VEC a^N)$. We let $\ind{\cdot}$ denote the Dirac probability measure over the point $\{\cdot\}$.


We now define push-forward of a measure \cite[pp. 118]{amb2008}, which we use throughout the paper.

\begin{definition}[Push-Forward of a measure \cite{amb2008}]
Let $\ALP A$ and $\ALP B$ be two separable metric spaces, $\mu\in\wp(\ALP A)$, and $r:\ALP A\rightarrow\ALP B$ a Borel measurable map. Then, a push-forward of $\mu$ through $r$, denoted by $r_{\#}\mu\in\wp(\ALP B)$, is defined as $r_{\#}\mu(\SF B):=\mu(r^{-1}(\SF B))$ for all Borel sets $\SF B\subset\ALP B$.{\hfill$\Box$}
\end{definition}
\begin{remark}
Let $\pj^{\ALP  A}:\ALP A\times\ALP B\rightarrow\ALP A$ be the projection map. Let $\mu\in\wp(\ALP A\times\ALP B)$ be a joint measure. Then, for any Borel set $\SF A\subset\ALP A$, $\pj^{\ALP  A}_{\#}\mu(\SF A):=\mu(\SF A\times\ALP B)$ is the marginal of the joint measure $\mu$.{\hfill$\Box$}
\end{remark}

\section{Problem Formulation}\label{sec:dynamicproblem}
Within a state-space model, we consider an $N$-agent dynamic team problem with no observation sharing information structure. Let the state at time instant $t\in [T]$ be denoted by $\VEC X_t$, and the space of all possible states at time $t$ be denoted by $\ALP X_t$. The action of Agent $i$ at time $t$ lies in a space $\ALP U^i_t$, and its action is denoted by $\VEC U^i_t$. The agents observe the plant states and past control actions through noisy sensors (or channels), and the observation of Agent $i$ at time $t$ is denoted by $\VEC Y^i_t$, which lies in a space $\ALP Y^i_t$. Throughout the paper, the spaces $\ALP X_t,\ALP U^i_t,\ALP Y^i_t$ are assumed to be complete separable metric spaces (also called Polish spaces, such as $\Re^n$ or the space of measures over $\Re^n$) at all time steps $t\in[T]$ and for all agents $i\in[N]$.

The state of the system evolves according to
\beq{\label{eqn:xt1}\VEC X_{t+1} =\tilde f_t(\VEC X_{1:t},\VEC U^{1:N}_{1:t},\VEC W^0_t),\qquad t\in [T],}
where $\VEC W^0_t$ is the actuation noise on the system. We denote the realization space of all possible actuation noises by $\ALP W^0_t$. Agent $i$ at time $t\in[T]$ makes an observation which depends on past states and control actions according to
\beq{\label{eqn:yit}\VEC Y^i_t = \tilde h^i_t(\VEC X_{1:t},\VEC U^{1:N}_{1:t-1},\VEC W^i_t),}
where $\VEC W^i_t$, which takes values in the space $\ALP W^i_t$, is the observation noise of Agent $i$ at time $t$. We again assume that $\ALP W^i_t, i\in\{0\}\cup[N], t\in[T]$ are Polish spaces. We make the following assumption on the state transition functions and the observation functions of the agents.
\begin{assumption}\label{ass:stf}
The state transition functions $\tilde f_t$ and observation functions $\tilde h^i_t, i\in [N],$ are continuous functions of their arguments for all time steps $t\in [T]$. {\hfill$\Box$}
\end{assumption}

The random variables $\{\VEC X_1,\VEC W^{0:N}_{1:T}\}$ are primitive random variables, and are assumed to be mutually independent. We let $\xi_{\ALP X_1}$ denote the probability measure on $\ALP X_1$ and $\xi_{\ALP W^i_t}$ denote the probability measure on $\ALP W^i_t$ for $i\in \{0\}\cup [N]$ and $t\in [T]$.

\subsection{Information Structures and Strategies of the Agents}
At each instant of time, we assume that the only information each agent acquires is its own observation, that is, $\VEC I^i_t:=\VEC Y^i_t$. Each agent uses its information to determine its control action. Toward this end, we allow the agents to act in a predetermined fashion, and when they act they either choose a deterministic strategy, or a randomized or behavioral strategy. We define these two notions of strategies below:
\begin{definition}[Deterministic Strategy]
A {\it deterministic strategy} for an Agent $i$ at time $t$ is a Borel measurable map $\gamma^i_t:\ALP Y^i_t\rightarrow\ALP U^i_t$. Let $\ALP D^i_t$ be the space of all such maps, which we call the {\it deterministic strategy space} of Agent $i\in[N]$ at time $t\in[T]$.
\end{definition}

\begin{definition}[Behavioral Strategy]
A {\it behavioral strategy} of Agent $i$ at time $t$ is a conditional measure $\pi^i_t$ satisfying the following two properties:
\begin{enumerate}
\item For every $\VEC y^i_t\in\ALP Y^i_t$, $\pi^i_t(\cdot|\VEC y^i_t)\in\wp(\ALP U^i_t)$;
\item For every $\SF U\in\FLD B(\ALP U^i_t)$, $\VEC y^i_t\mapsto \pi^i_t(\SF U|\VEC y^i_t)$ is a $\FLD B(\ALP Y^i_t)$-measurable function.
\end{enumerate}
Let $\ALP R^i_t$ denote the behavioral strategy space of Agent $i\in[N]$ at time $t\in[T]$.{\hfill $\Box$}
\end{definition}

\begin{remark}
For an agent at any time step, any deterministic strategy is by definition also a behavioral strategy. For example, if $\gamma^i_t$ is a deterministic strategy of Agent $i$ at time $t$, then the corresponding (induced) behavioral strategy is $\pi^i_t(d\VEC u^i_t|\VEC y^i_t) = \ind{\gamma^i_t(\VEC y^i_t)}(d\VEC u^i_t)$. Thus, the set of behavioral strategies of an agent subsumes the set of deterministic strategies of that agent.{\hfill$\Box$}
\end{remark}

As a consequence of the remark above, throughout this paper, we will work with behavioral strategies of an agent with the understanding that this also covers deterministic strategies of that agent as well.

\subsection{Expected Cost Functional of the Team}
The team is equipped with a cost function $\tilde{c}$, which is assumed to be a non-negative continuous function of all states $\VEC X_{1:T+1}$, observations $\VEC Y^{1:N}_{1:T}$ and actions $\VEC U^{1:N}_{1:T}$ of the agents. However, we can substitute \eqref{eqn:xt1} recursively so that the cost function becomes purely a function of the primitive random variables $\{\VEC X_1,\VEC W^0_{1:T}\}$, observations $\VEC Y^{1:N}_{1:T}$ and the control actions of all agents, $\VEC U^{1:N}_{1:T}$. Therefore, for a fixed realization of the primitive random variables, observations and control actions of the agents, the cost incurred by the team can be written as $c(\VEC x_1,\VEC w^0_{1:T},\VEC y^{1:N}_{1:T},\VEC u^{1:N}_{1:T})$ for some $c$, which is clearly related to $\tilde c$. We have the following result on the cost function $c$.
\begin{lemma}
If Assumption \ref{ass:stf} holds, then the cost function $c$ is continuous.
\end{lemma}
\begin{proof}
Note that by construction, $c$ is generated by $\tilde c$ as 
\beqq{c(\VEC x_1,\VEC w^0_{1:T},\VEC y^{1:N}_{1:T},\VEC u^{1:N}_{1:T}) = \tilde c(\VEC x_1,\tilde f_1(\VEC x_1,\VEC u^{1:N}_1,\VEC w^0_1),\cdots,\tilde f_T(\VEC x_{1:T},\VEC u^{1:N}_{1:T},\VEC w^0_T), \VEC y^{1:N}_{1:T},\VEC u^{1:N}_{1:T}),}
where $\VEC x_t$ is substituted as a function of $\VEC x_1,\VEC w^0_{1:t-1}$ and $\VEC u^{1:N}_{1:t-1}$ using \eqref{eqn:xt1} for all $t\in[T]$. Since $\tilde c$ and $\{f_t\}_{t\in[T]}$ are continuous functions of their arguments, we conclude that $c$ is a continuous function on $\ALP X_1\times\ALP W^0_{1:T}\times\ALP Y^{1:N}_{1:T}\times\ALP U^{1:N}_{1:T}$.
\end{proof}

In standard optimal control problems, the cost function of the team is taken to be a sum of stage-wise cost functions, in which the cost function at every time step depends on the current state and actions of the agents. However, we do not assume such a structure on the cost function of the team problem considered in this paper. This general cost function encompasses ones that appear in certain classes of communication systems, economic systems, and feedback control over noisy channels.

Throughout this paper, we use $J:\ALP R^{1:N}_{1:T}\rightarrow\Re_+$ to denote the expected cost functional of the team, which is defined as
\beqq{J(\pi^{1:N}_{1:T}) = \ex{c(\VEC x_1,\VEC w^0_{1:T},\VEC y^{1:N}_{1:T},\VEC u^{1:N}_{1:T})},}
where the expectation is taken with respect to the measure induced on the random variables by the choice of behavioral strategies $\pi^{1:N}_{1:T}$. We make the following natural assumption on the team problem described above.
\begin{assumption}\label{ass:one}
There exists a set of behavioral strategies $\tilde{\pi}^{1:N}_{1:T}\in\ALP R^{1:N}_{1:T}$ of the agents, which results in finite expected cost to the team.{\hfill $\Box$}
\end{assumption}

\subsection{Solution Approach and the Proof Program}
Our proof of existence of optimal strategies in team problems formulated above follows the following steps:

\begin{enumerate}
\item We first show the existence of optimal strategies in a static team problem in Section \ref{sec:staticteam}, in which (i) the cost function of the team is  continuous and bounded function of its arguments, (ii) the action spaces of the agents are compact, and (iii) the observation channels of the agents satisfy a technical assumption. We refer to this static team problem as Team {\bf ST1}. We establish a tightness result on the joint measures over state, observation and action spaces of each agent. For any sequence of joint measures induced by behavioral strategies of the agents that achieves expected costs converging to the infimum of the expected cost of the team, we show that there exists a convergent subsequence of joint measures, which are induced by a set of behavioral strategies, whose limit achieves the infimum of the expected cost functional of the team.

\item Next, we show the existence of optimal strategies in a static team problem in which (i) the action spaces of the agents are non-compact and (ii) cost function of the team is continuous and has a coercive structure. We refer to this static team problem as Team {\bf ST2}. This result is established in Section \ref{sec:unbounded} using the results of Section \ref{sec:staticteam}.

\item Subsequently, we use Witsenhausen's static reduction technique to reduce a sequential dynamic team into a (reduced) static team with independent observations of the agents. We assume a certain structure on the cost function of the dynamic team. A challenge with this approach is that the cost function of the reduced static team problem may not satisfy a coercivity condition (introduced later) even though the cost function of the dynamic team satisfies that coercivity condition. To alleviate this problem, we provide a novel approach by restricting the search for optimal behavioral strategies to a compact set. The dynamic team problem is then solved under some mild assumptions that are delineated in Section \ref{sec:staticreduction}.
\end{enumerate}

In the next section, we establish the existence of team-optimal solutions to the first static team problem mentioned above.

\section{Existence of Optimal Solution in {\bf ST1}}\label{sec:staticteam}
In this section, we study the $N$-agent static team problem in which each Agent $i$ observes a random variable $\VEC Y^i$, correlated with the random variable $\VEC X$, and takes an action $\VEC U^i$. We let $\ALP X$ denote the state space of the team, and $\ALP Y^i$ and $\ALP U^i$ denote, respectively, the observation space and action space of Agent $i$. We assume that $\ALP U^i$ is a compact subset of a Polish space for all $i\in[N]$.

The team incurs a cost $c$, which is a non-negative continuous function of the state, observations and the control actions of all the agents, that is, $c:\ALP X\times\ALP Y^{1:N}\times\ALP U^{1:N}\rightarrow\Re^+$. The expected cost functional of the team, denoted by $J:\ALP R^{1:N}\rightarrow\Re^+$, as a function of behavioral strategies of the agents, is
\beqq{J(\pi^{1:N}) = \int_{\ALP X\times \ALP U^{1:N}}\int_{\ALP Y^{1:N}} c(\VEC x,\VEC y^{1:N},\VEC u^{1:N}) \prod_{i=1}^N\pi^i(d\VEC u^i|\VEC y^i) \pr{d\VEC x,d\VEC y^1,\ldots,d\VEC y^N}.}

We show that, under certain conditions, there exists a $\pi^{1:N\star}\in\ALP R^{1:N}$ such that
\beqq{J(\pi^{1:N\star}) = \inf_{\pi^{1:N}\in\ALP R^{1:N}} J(\pi^{1:N}).}

We first provide an outline of our approach to showing the existence of optimal strategies in the static team problem. Consider a sequence $\{\pi^{1:N}_n\}_{n\in\Na}\subset\ALP R^{1:N}$ of control strategies of the agents such that $\lf{n}J(\pi^{1:N}_n)  = \inf_{\pi^{1:N}\in\ALP R^{1:N}} J(\pi^{1:N})$. There are three issues that need to be resolved: The first issue is that the sequence of joint measures $\Big\{\prod_{i=1}^N\pi^i_n(d\VEC u^i|\VEC y^i) \pr{d\VEC x,d\VEC y^1,\ldots,d\VEC y^N}\Big\}_{n\in\Na}$ may not be a weak* convergent sequence. This can be remedied by considering a convergent subsequence of $\Big\{\prod_{i=1}^N\pi^i_n(d\VEC u^i|\VEC y^i) \pr{d\VEC x,d\VEC y^1,\ldots,d\VEC y^N}\Big\}_{n\in\Na}$. The second problem is to ensure that the limit of the convergent subsequence satisfies the informational constraint. This means that the conditional measure on the action space of Agent $i$ given the observation of that agent and the state of the limiting measure must be independent of the state for 
any $i\in[N]$. The third problem is that if for all $i\in[N]$, $\{\pi^i_n(d\VEC u^i|\VEC y^i)\pr{d\VEC x,d\VEC y^i}\}_{n\in\Na}$ converges in the weak* sense to a measure $\pi^i_0(d\VEC u^i|\VEC y^i)\pr{d\VEC x,d\VEC y^i}$ for some $\pi^i_0\in\ALP R^i$,  then the expected cost functional $J$ may not satisfy $\lf{n}J(\pi^{1:N}_n) = J(\pi^{1:N}_0)$. We overcome all these three challenges by employing the following steps:

\begin{enumerate}
\item We show that for any $g\in U_b(\ALP X\times\ALP Y^{1:N}\times\ALP U^{1:N})$, $\left\{\int g\:\pi^i_n(d\VEC u^i|\VEC y^i)\pr{d\VEC y^i|\VEC x}\right\}_{n\in\Na}$ is a uniformly equicontinuous and bounded sequence of functions  under some assumptions on the conditional measure $\pr{d\VEC y^i|\VEC x}$.
\item In order to satisfy the informational constraint of the limiting measure of any convergent subsequence of the sequence $\{\pi^i_n(d\VEC u^i|\VEC y^i)\pr{d\VEC x,d\VEC y^i}\}_{n\in\Na}$, we assume a specific structure on the conditional probability measure $\pr{d\VEC x|\VEC y^i}$.
\item We extract a weak* convergent subsequence
\beqq{\left\{\prod_{i=1}^N\pi^{i}_{n_k}(d\VEC u^i|\VEC y^i)\pr{d\VEC x,d\VEC y^1,\ldots,d\VEC y^N}\right\}_{k\in\Na}}
of the sequence of measures such that $\lf{k}J(\pi^{1:N}_{n_k})= J(\pi^{1:N}_0)$.
\item Once we show that there exists a set of behavioral strategies of the agents that achieves the minimum expected cost, we use Blackwell's irrelevant information theorem \cite[p. 457]{yukselbook} to conclude that there exists a set of deterministic strategies of the agents that achieves the minimum expected cost.
\end{enumerate}

In order to prove existence, we require several auxiliary results that are proved in the next subsection. The purpose of the auxiliary results is to establish results 1 and 2 above in a somewhat general setting. The existence of optimal strategies is proven in Subsection \ref{sub:existence}. We also show in that subsection the existence of optimal strategies in static teams when the state is degenerate and the observations of the agents are mutually independent random variables.

\begin{remark}
Without the further regularity conditions to be presented shortly, the approach above may not be sufficient to lead to the desired existence result for teams in which the information of the agents given the state are not conditionally independent. This is the case, for example, when observations are shared by the agents in a team. In Appendix \ref{app:counterexample}, a counterexample of a two-agent static team is given, in which one agent shares its observation with another agent. We show that even if sequences of the joint measures of each agent's action and information (which may include another agent's observation too) converge in the weak* sense, the corresponding sequence of joint measures over the product of all action, observation, and state spaces need not converge. This counterexample shows that additional regularity conditions on the cost function, observation channels, underlying distributions of primitive random variables, and topologies on measure spaces are needed to establish the existence of 
optimal strategies in general static and dynamic teams where observations may be shared.
{\hfill $\Box$}
\end{remark}

\subsection{Auxiliary Results}
In this subsection, we state a few lemmas that are needed to prove the existence of optimal strategies in the static team problem formulated above. Unless otherwise stated, $\ALP A$, $\ALP B^i$, $\ALP C$ and $\ALP Y^i, \: i\in[N]$ denote Polish spaces, with generic elements in these spaces denoted, respectively, by $\VEC a$, $\VEC b^i$, $\VEC c$ and $\VEC y^i$. We now introduce a condition on the conditional probability measures, which will be important in proving the auxiliary results.


\begin{definition}[Condition {\bf C1}]\label{def:obskernel}
Let $\VEC A$ and $\VEC Y$ be random variables such that $\pr{d\VEC y|\VEC a} = \eta(\VEC a,\VEC y)\nu(d\VEC y)$ for some non-negative measure $\nu\in\ca(\ALP Y)$. We say that the pair $(\eta,\nu)$ satisfies condition {\bf C1} if and only if
\begin{enumerate}
\item $\eta$ is a continuous function of its arguments, that is, $\eta\in C(\ALP A\times\ALP Y)$; and
\item there exists a bounded measurable function $h_{(\eta,\nu)}:\ALP A\times\ALP Y\rightarrow\Re^+$ satisfying $\sup_{\VEC a\in\ALP A} \int h_{(\eta,\nu)}d\nu <\infty$, such that for every $\epsilon>0$, there exists a $\delta>0$ such that for any $\VEC a_0\in\ALP A$ and for all $\VEC a\in\ALP A$ satisfying $d_{\ALP A}(\VEC a,\VEC a_0)<\delta$, we have
\beqq{|\eta(\VEC a,\VEC y)-\eta(\VEC a_0,\VEC y)|<\epsilon \:h_{(\eta,\nu)}(\VEC a_0,\VEC y).}
\end{enumerate}
We call $h_{(\eta,\nu)}$ the variation control (VC) function of the pair $(\eta,\nu)$. {\hfill $\Box$}
\end{definition}


If the observation is an additive noise corrupted version of the state, with the noise being Gaussian, then the above condition holds. This is noted in the following example.

\begin{example}
Let $\ALP A=\ALP Y=\Re^n$. A sufficient condition for a pair $(\eta,\nu)$ to satisfy condition {\bf C1} is $\VEC Y = \VEC A+\VEC W$ for some Gaussian random vector $\VEC W$ with density function $N(\cdot)$ and a positive definite covariance. In this case, $\nu$ is the usual Lebesgue measure on $\Re^n$, $\eta(\VEC a,\VEC y) = N(\VEC y-\VEC a)$, and the VC function for the pair $(\eta,\nu)$ is
\[h_{(\eta,\nu)}(\VEC a_0,\VEC y):=\max_{\VEC a\in \ball{\VEC a_0,1}}\bigg\|\frac{d\eta}{d\VEC a}\bigg\|_2.\]
Since $\frac{d\eta}{d\VEC a}$ decays exponentially as $\|\VEC y\|_2\rightarrow\infty$, $h_{(\eta,\nu)}$ has the property that $\int h_{(\eta,\nu)}d\nu <\infty$ and it is a constant function over $\ALP A$. {\hfill $\Box$}
\end{example}


Now, we make use of the uniform continuity of a function and Condition {\bf C1} on the conditional measure to prove the following result.

\begin{lemma}\label{lem:unifequi}
Let $\nu\in\ca(\ALP Y)$ be a non-negative measure and $\mu\in\wp(\ALP Y\times\ALP B)$. Let $\VEC A$, $\VEC B$, $\VEC C$ and $\VEC Y$ be random variables such that $\pr{d\VEC y|\VEC a} = \eta(\VEC a,\VEC y)\nu(d\VEC y)$, where the pair $(\eta,\nu)$ satisfies condition {\bf C1} with VC function $h_{(\eta,\nu)}$. Let $g:\ALP A\times\ALP B\times\ALP C\rightarrow\Re$ be a uniformly continuous and bounded function. Define the map $f:\ALP A\times\ALP C\rightarrow\Re$ by
\beqq{f_\mu(\VEC a,\VEC c) = \int_{\ALP B} g(\VEC a,\VEC b,\VEC c) \mu(d\VEC b|\VEC y)\eta(\VEC a,\VEC y) \nu(d\VEC y).}
Then, $\{f_\mu(\cdot,\cdot)\}_{\mu\in\wp(\ALP Y\times\ALP B)}$ is a uniformly equicontinuous and bounded map on its domain.
\end{lemma}
\begin{proof}
See Appendix \ref{app:unifequi}.
\end{proof}

We now have a corollary to this result.

\begin{corollary}\label{cor:fconti}
Under the same assumptions and notation as in Lemma \ref{lem:unifequi}, for every $\mu\in\wp(\ALP Y\times\ALP B)$, $f_{\mu}$ is a uniformly continuous function on its domain.
\end{corollary}

We now state another important lemma, whose proof is similar to the proof of Lemma \ref{lem:unifequi}.
\begin{lemma}\label{lem:unifequi2}
Let $\nu^i\in\ca(\ALP Y^i)$ be a non-negative measure and $\mu^i\in\wp(\ALP Y^i\times\ALP B^i)$, $i\in[N]$. Let $\VEC A$, $\VEC B^i$, $\VEC C$ and $\VEC Y^i$ be random variables for $i\in[N]$ such that $\pr{d\VEC y^i|\VEC a} = \eta^i(\VEC a,\VEC y^i)\nu^i(d\VEC y^i)$, where each pair $(\eta^i,\nu^i)$ satisfies condition {\bf C1} with VC function $h^i$. Further, assume that random variables $\VEC Y^1,\ldots,\VEC Y^N$ are independent given $\VEC a$. Let $g\in U_b(\ALP A\times\ALP B^{1:N}\times\ALP C)$. Define the map $f:\ALP A\times\ALP C\rightarrow\Re$ by
\beqq{f_{\mu^{1:N}}(\VEC a,\VEC c) = \int_{\ALP B^{1:N}} g(\VEC a,\VEC b^{1:N},\VEC c) \prod_{i=1}^N\mu^i(d\VEC b^i|\VEC y^i)\eta^i(\VEC a,\VEC y^i) \nu^i(d\VEC y^i).}
Then, $\{f_{\mu^{1:N}}(\cdot,\cdot)\}_{\mu^i\in\wp(\ALP Y^i\times\ALP B^i)}$ is a uniformly equicontinuous and bounded map on its domain.
\end{lemma}
\begin{proof}
See Appendix \ref{app:corunifequi}.
\end{proof}

The result of Lemma \ref{lem:unifequi2} allows us to apply Arzela-Ascoli Theorem \cite{ali2006} on compact subsets of the domain to obtain a convergent subsequence that converges to some bounded continuous function pointwise (not in sup norm). We then need the following result.

\begin{lemma}\label{lem:lambdan}
Under the same assumptions and notation as in Lemmas \ref{lem:unifequi} and \ref{lem:unifequi2}, let $\{\mu^i_n\}_{n\in\Na}\subset\wp(\ALP Y^i\times\ALP B^i)$ be an arbitrary sequence of measures for $i\in[N]$. For every $n\in\Na$, define $f_n:= f_{\mu^{1:N}_n}$. Further, assume that $\ALP A,\ALP C,\ALP Y^i,\ALP B^i,i\in[N]$ are all $\sigma$-compact Polish spaces. If $\{\zeta_n\}_{n\in\Na}\subset\wp(\ALP A\times\ALP C)$ is a weak* convergent sequence of measures converging to $\zeta_0$, then there exists a subsequence $\{n_k\}_{k\in\Na}$ such that
\beqq{\lf{k}\bigg|\int_{\ALP A\times\ALP C} f_{n_k}d\zeta_{n_k} - \int_{\ALP A\times\ALP C} f_{n_k}d\zeta_{0}\bigg| = 0 .}
\end{lemma}
\begin{proof}
See Appendix \ref{app:lambdan}.
\end{proof}

We have thus stated (and proved) all the major auxiliary results that are needed to establish the existence of optimal strategies. We next prove an additional result, which states that under some sufficient condition, if we take a weak* convergent sequence of measures satisfying a conditional independence property, then the limit also satisfies the conditional independence property. This result is useful to show that the weak* convergent sequence of joint measures over the state, observation and action of Agent $i$ does not converge to a limit in which the control action depends on both the state and the observation.

\begin{lemma}\label{lem:control2}
Let $\{\mu_n\}_{n\in\Na}\subset\wp(\ALP A\times\ALP B\times\ALP C)$ be a convergent sequence of measures such that $\mu_n(d\VEC a,d\VEC b,d\VEC c) = \mu_n(d\VEC c|\VEC b)\zeta(d\VEC a,d\VEC b)$, where $\zeta\in \wp(\ALP A\times\ALP B)$ with the property that $\zeta(d\VEC a|\VEC b) = \rho(\VEC a,\VEC b)\nu(d\VEC a)$ for some $\rho\in C_b(\ALP A\times\ALP B)$ and non-negative measure $\nu$ on $\ALP A$. Assume that $(\rho,\nu)$ satisfies Condition {\bf C1}. If $\mu_n\ws\mu_0$ for some $\mu_0\in\wp(\ALP A\times\ALP B\times\ALP C) $, then $\mu_0(d\VEC a,d\VEC b,d\VEC c) = \mu_0(d\VEC c|\VEC b)\zeta(d\VEC a,d\VEC b)$.
\end{lemma}
\begin{proof}
See Appendix \ref{app:control2}.
\end{proof}

In the next subsection, we turn our attention to proving the existence of optimal strategies for the static team problem considered in this section.

\subsection{Existence of Optimal Strategies}\label{sub:existence}
In this subsection, we prove one of the main results of the paper. We make the following assumption on the probability measure on $\ALP X\times\ALP Y^{1:N}$.
\begin{assumption}\label{assm:condind}
The spaces $\ALP X,\ALP Y^i$ and $\ALP U^i$ are $\sigma$-compact Polish spaces for all $i\in[N]$. Further, there exist bounded continuous functions $\eta^i$ and $\rho^i$, $i\in[N]$ such that
\beqq{\pr{d\VEC y^{1:N}|\VEC x} &=& \prod_{i\in[N]}\pr{d\VEC y^{i}|\VEC x},\\
\pr{d\VEC y^i|\VEC x} &=& \eta^i(\VEC x,\VEC y^i)\nu_{\ALP Y^i}(d\VEC y^i),\\
\pr{d\VEC x|\VEC y^i} &=& \rho^i(\VEC y^i,\VEC x)\nu_{\ALP X}(d\VEC x),}
where $\nu_{\ALP X}$ and $\nu_{\ALP Y^i}$ are measures on $\ALP X$ and $\ALP Y^i$, respectively, for $i\in[N]$. The tuples $(\rho^i,\nu_{\ALP X})$ and $(\eta^i,\nu_{\ALP Y^i}),\: i\in[N]$ satisfy Condition {\bf C1}.{\hfill $\Box$}
\end{assumption}

We now use the auxiliary results in the previous subsection to prove the following important theorem.
\begin{theorem}\label{thm:intc}
Consider Team {\bf ST1}, satisfying Assumption \ref{assm:condind}, where $\ALP U^i, i\in[N]$ need not be compact sets. Let $g\in U_b(\ALP X\times\ALP Y^{1:N}\times\ALP U^{1:N})$. For every $i\in[N]$, let $\{\lambda^i_n\}_{n\in\Na}\subset\wp(\ALP U^i\times\ALP Y^i\times\ALP X)$ be a convergent sequence of measures such that $\lambda^i_n(d\VEC u^i,d\VEC y^i,d\VEC x) = \lambda^i_n(d\VEC u^i|\VEC y^i)\pr{d\VEC x,d\VEC y^i} $, converging to some $\lambda^i_0\in \wp(\ALP U^i\times\ALP Y^i\times\ALP X)$. Then,
\beqq{\lf{n}\int g(\VEC x,\VEC y^{1:N},\VEC u^{1:N}) \left(\prod_{i=1}^N \lambda^i_{n}(d\VEC u^i|\VEC y^i)\right)\pr{d\VEC x,d\VEC y^{1:N}}\\
= \int g(\VEC x,\VEC y^{1:N},\VEC u^{1:N}) \left(\prod_{i=1}^N \lambda^i_0(d\VEC u^i|\VEC y^i)\right)\pr{d\VEC x,d\VEC y^{1:N}}.}
Consequently, we have
\beqq{\left(\prod_{i=1}^N \lambda^i_{n}(d\VEC u^i|\VEC y^i)\right)\pr{d\VEC x,d\VEC y^{1:N}}\ws \left(\prod_{i=1}^N \lambda^i_0(d\VEC u^i|\VEC y^i)\right)\pr{d\VEC x,d\VEC y^{1:N}} \quad \text{ as }n\rightarrow\infty.}
\end{theorem}
\begin{proof}
See Appendix \ref{app:intc}.
\end{proof}

We now list in the following assumption the conditions that we need in order to establish the existence of optimal strategies.
\begin{assumption}\label{assm:staticteam}
\begin{enumerate}
\item The cost function $c:\ALP X\times\ALP Y^{1:N}\times\ALP U^{1:N}\rightarrow\Re^+$ is continuous in its arguments and bounded from above.
\item The action sets $\ALP U^i,i\in[N]$, of all the agents are compact subsets of Polish spaces. Therefore, $\ALP U^i$ is $\sigma$-compact Polish space for all $i\in[N]$.
\item Assumption \ref{assm:condind} holds.
\end{enumerate}
\end{assumption}

\begin{remark}
If we assume that the cost function $c$ is continuous and the spaces $\ALP X$ and $\ALP Y^i, i\in[N]$ are compact subsets of Polish spaces, then the cost function is bounded.{\hfill $\Box$}
\end{remark}

\begin{remark}
It should also be noted that Part 3 of Assumption \ref{assm:staticteam} is satisfied if (i) $\ALP X$ and $\ALP Y^i,\; i\in[N]$ are finite dimensional Euclidean spaces, and (ii) the state and observations are jointly Gaussian random variables such that Agent $i$ observes a Gaussian noise corrupted version of the state $\VEC X$.{\hfill $\Box$}
\end{remark}

The following theorem states that any team problem that satisfies the assumptions made above admits a team-optimal solution.
\begin{theorem}\label{thm:staticteam}
Every static team problem satisfying Assumption \ref{assm:staticteam} admits a team-optimal solution in deterministic strategies.
\end{theorem}
\begin{proof}
Let $\{\pi^{1:N}_n\}_{n\in\Na}\subset\ALP R^{1:N}$ be a sequence of strategy profiles of the agents such that
\beqq{J(\pi^{1:N}_n)<\inf_{\pi^{1:N}\in\ALP R^{1:N}}J(\pi^{1:N})+\frac{1}{n}.}
We next show that there exists a convergent subsequence of this sequence $\{\pi^{1:N}_n\}_{n\in\Na}$ such that the limiting behavioral strategies of the agents achieve the infimum of the expected cost functional. We organize the proof into four steps.

{\it Step 1: (Tightness)} Since $\ALP X$ and $\ALP Y^i$ are Polish spaces, $\pr{d\VEC x,d\VEC y^i}$ is a tight measure. Since $\ALP U^i$ is compact, the set of measures $\{\pi^i(d\VEC u^i|\VEC y^i)\pr{d\VEC x,d\VEC y^i}\}_{\pi^i\in\ALP R^i}$ is tight for all $i\in[N]$. Define a measure $\lambda^i_n$ as
\beqq{\lambda^i_n(d\VEC u^i,d\VEC y^i,d\VEC x) := \pi^i_n(d\VEC u^i|\VEC y^i)\pr{d\VEC x,d\VEC y^i},\quad \text{ for }n\in\Na,\; i\in[N].}

{\it Step 2: (Extracting convergent subsequence)} Recall that every sequence of tight measures has a convergent subsequence by Prohorov's theorem \cite{durrett2010}. Thus, $\{\lambda^1_n\}_{n\in\Na}$ must have a convergent subsequence, say $\{\lambda^1_{n_k}\}_{k\in\Na}$. Similarly, $\{\lambda^2_{n_k}\}_{k\in\Na}$ must have a convergent subsequence. Since there is only a finite number ($N$) of agents, we continue this process of extracting convergent subsequences of every sequence of measures to get a convergent subsequence of a set of measures $\{\lambda^1_{n_l},\ldots,\lambda^N_{n_l}\}_{l\in\Na}$ such that
\beqq{\lambda^i_{n_l}\ws\lambda^i_0,\quad\text{ as $l\rightarrow\infty$ for all } i\in[N],}
for some set of measures $\{\lambda^1_0,\ldots,\lambda^N_0\}$. Recall the result of Lemma \ref{lem:control2}, which implies that $\lambda^i_0(d\VEC u^i,d\VEC y^i,d\VEC x) = \lambda^i_0(d\VEC u^i|\VEC y^i)\pr{d\VEC x,d\VEC y^i} $. Define $\pi^i_0(d\VEC u^i|\VEC y^i):=\lambda^i_0(d\VEC u^i|\VEC y^i)$ for all $i\in[N]$ and $\pi^{1:N}_0 :=\{\pi^1_0,\ldots,\pi^N_0\}$.

{\it Step 3: (Limit achieves infimum)} The result of Theorem \ref{thm:intc} implies that
\beqq{\lf{l}J(\pi^{1:N}_{n_{l}}) = J(\pi^{1:N}_0) = \inf_{\pi^{1:N}\in\ALP R^{1:N}}J(\pi^{1:N}).}

{\it Step 4: (Applying a result on irrelevant information due to Blackwell) \cite{blackwell1963,blackwell1964}} Now, using Blackwell's irrelevant information theorem (see e.g. \cite[p. 457]{yukselbook}), we conclude that for fixed optimal behavioral strategies of all agents other than $i$, there exists a deterministic strategy of Agent $i$ that achieves the same value of expected cost as the optimal behavioral strategy of Agent $i$. Thus, all agents' strategies can be restricted, without any loss of generality, to deterministic ones. This implies that the static team admits an optimal solution in the class of deterministic strategies of the agents. This completes the proof of the theorem.
\end{proof}

\begin{remark}
It must be noted that for the existence result proven above, we do not require the state space $\ALP X$ and observation spaces $\ALP Y^i, i\in[N]$ to be compact. {\hfill $\Box$}
\end{remark}


We have an immediate corollary.
\begin{corollary}\label{cor:staticteam}
Assume that every Agent $i\in [N]$ observes $\VEC Y^i$, which is independent of the observations of all other agents. If the cost function is continuous in its arguments (observations and actions of the agents) and bounded, and action spaces of the agents are compact, then the static team with independent observations admits an optimal solution in deterministic strategies.
\end{corollary}


\subsection{Static Team {\bf ST1} with Degraded Information}\label{sub:degraded}
In Assumption \ref{assm:condind}, we assumed conditional independence of observations given the state, which we relax in this subsection. For simplicity, we consider a two-agent static team problem, where the observation of Agent 2 is a noise corrupted version of the observation of Agent 1. We further invoke the following assumption. The main results of Lemma \ref{lem:intcdegraded} and Theorem \ref{thm:degradedstatic} below, and the main idea of the proof can be extended to multi-agent static team scenarios.

\begin{assumption}\label{assm:degraded}
Consider Team {\bf ST1} in which Agent 2 observes a degraded version of Agent 1's observation. The spaces $\ALP X,\:\ALP Y^i$ and  $\:\ALP U^i$ are $\sigma$-compact Polish spaces for all $i\in\{1,2\}$. There exist bounded continuous functions $\eta^i$ and $\rho^i$, $i\in\{1,2\}$ such that
\beqq{\pr{d\VEC y^2,d\VEC y^1|\VEC x} &=& \pr{d\VEC y^2|\VEC y^1}\pr{d\VEC y^1|\VEC x},\\
\pr{d\VEC y^2|\VEC y^1} &=& \eta^2(\VEC y^1,\VEC y^2)\nu_{\ALP Y^2}(d\VEC y^2),\\
\pr{d\VEC y^1|\VEC x,\VEC y^2} &=& \eta^1(\VEC x,\VEC y^2,\VEC y^1)\nu_{\ALP Y^1}(d\VEC y^1),\\
\pr{d\VEC y^1,d\VEC x|\VEC y^2} &=& \rho^2(\VEC y^2,\VEC y^1,\VEC x)\nu_{\ALP Y^1}(d\VEC y^1)\nu_{\ALP X}(d\VEC x),\\
\pr{d\VEC x|\VEC y^1} &=& \rho^1(\VEC y^1,\VEC x)\nu_{\ALP X}(d\VEC x),}
where $\nu_{\ALP X}$ and $\nu_{\ALP Y^i}$ are measures on $\ALP X$ and $\ALP Y^i$, respectively, for $i\in\{1,2\}$. The tuples $(\rho^1,\nu_{\ALP X})$, $(\rho^2,\nu_{\ALP X}\times\nu_{\ALP Y^1})$ and $(\eta^i,\nu_{\ALP Y^i}),\: i\in\{1,2\}$ satisfy Condition {\bf C1}.{\hfill $\Box$}
\end{assumption}

We now use the auxiliary results in the previous subsection to prove the following important lemma.
\begin{lemma}\label{lem:intcdegraded}
Consider Team {\bf ST1} satisfying Assumption \ref{assm:degraded}, where $\ALP U^i, i\in\{1,2\}$ need not be compact sets. Let $\ALP X,\:\ALP Y^i,\:\ALP U^i$, $i\in\{1,2\}$ be $\sigma$-compact Polish spaces. Let $g\in U_b(\ALP X\times\ALP Y^{1:2}\times\ALP U^{1:2})$. For every $i\in\{1,2\}$, let $\{\lambda^1_n\}_{n\in\Na}\subset\wp(\ALP U^i\times\ALP Y^1\times\ALP X)$ and $\{\lambda^2_n\}_{n\in\Na}\subset\wp(\ALP U^i\times\ALP Y^1\times\ALP Y^2\times\ALP X)$ be convergent sequences of measures such that $\lambda^1_n(d\VEC u^1,d\VEC y^1,d\VEC x) = \lambda^1_n(d\VEC u^1|\VEC y^1)\pr{d\VEC x,d\VEC y^1} $ and $\lambda^2_n(d\VEC u^2,d\VEC y^{1:2},d\VEC x) = \lambda^2_n(d\VEC u^2|\VEC y^2)\pr{d\VEC x,d\VEC y^{1:2}} $, converging to some $\lambda^1_0\in \wp(\ALP U^1\times\ALP Y^1\times\ALP X)$ and $\lambda^2_0\in \wp(\ALP U^2\times\ALP Y^1\times\ALP Y^2\times\ALP X)$, respectively. Then,
\beqq{\lf{n}\int g(\VEC x,\VEC y^{1:2},\VEC u^{1:2}) \left(\prod_{i=1}^2 \lambda^i_{n}(d\VEC u^i|\VEC y^i)\right)\pr{d\VEC x,d\VEC y^{1:2}}\\
= \int g(\VEC x,\VEC y^{1:2},\VEC u^{1:2}) \left(\prod_{i=1}^2 \lambda^i_0(d\VEC u^i|\VEC y^i)\right)\pr{d\VEC x,d\VEC y^{1:2}}.}
Consequently, we have
\beqq{\left(\prod_{i=1}^2 \lambda^i_{n}(d\VEC u^i|\VEC y^i)\right)\pr{d\VEC x,d\VEC y^{1:2}}\ws \left(\prod_{i=1}^2 \lambda^i_0(d\VEC u^i|\VEC y^i)\right)\pr{d\VEC x,d\VEC y^{1:2}} \quad \text{ as }n\rightarrow\infty.}
\end{lemma}
\begin{proof}
See Appendix \ref{app:intcdegraded}.
\end{proof}

We can now show the existence of optimal strategies for a static team {\bf ST1} satisfying Parts 1 and 2 of Assumption \ref{assm:staticteam}, and Assumption \ref{assm:degraded}.

\begin{theorem}\label{thm:degradedstatic}
Any two-agent static team {\bf ST1} satisfying Parts 1 and 2 of Assumption \ref{assm:staticteam}, and Assumption \ref{assm:degraded} admits a team-optimal solution, which is in the class of deterministic strategies of the agents.
\end{theorem}
\begin{proof}
The proof follows by mimicking the steps of the proof of Theorem \ref{thm:staticteam} and using Lemma \ref{lem:intcdegraded}.
\end{proof}

In the above theorem, we showed that an assumption of conditional independence of observations given the state is not needed for the existence of optimal strategies in team problems; we considered a case where the observation of one agent is a degraded version of the observation of another agent and showed that optimal strategies exist under certain assumptions.

We note that our setting does not cover teams with observation sharing information structures, because the technique we employed for proving Theorem \ref{thm:staticteam} does not readily carry over to such teams. In particular, if the agents share their observations in a certain manner, then we cannot show the equicontinuity result of Lemma \ref{lem:unifequi2}, which is used to prove Theorem \ref{thm:intc}. Recall that Theorem \ref{thm:intc} is crucial for the proof of Theorem \ref{thm:staticteam}. See Appendix \ref{app:counterexample} for a detailed discussion and a counterexample. However, for a class of problems with observation sharing information pattern, it is possible to use other techniques, such as dynamic programming or viewing the decision makers with common information as a single decision maker. We leave a systematic analysis of this setup to future work.

This concludes the discussion in this section. In the next section, we extend the ideas developed in this section to obtain sufficient conditions on a team problem with non-compact action spaces and unbounded continuous cost function, for existence of a team-optimal solution.

\section{Existence of Optimal Solution in {\bf ST2}}\label{sec:unbounded}
In this section, we consider the static team problem in which the cost function is non-negative, continuous, but may be unbounded, and the action sets may be non-compact. We build on the results proved in the previous section to show the existence of optimal strategies of agents in such a team problem.

In the next subsection, we use the result from Theorem \ref{thm:intc} to investigate the properties of the expected cost functional, as a function of the behavioral strategies of the agents, of the team problem with unbounded cost and non-compact action spaces.
\subsection{Properties of the Expected Cost Functional}
Our first result uses Theorem \ref{thm:intc} to prove an important property of expected cost functional of team {\bf ST2}.
\begin{theorem}\label{thm:sub}
Recall that $c:\ALP X\times\ALP Y^{1:N}\times\ALP U^{1:N}\rightarrow\Re^+$ is a non-negative continuous function. For every $i\in[N]$, let $\{\lambda^i_n\}_{n\in\Na}\subset\wp(\ALP U^i\times\ALP Y^i\times\ALP X)$ be a convergent sequence of measures such that $\lambda^i_n(d\VEC u^i,d\VEC y^i,d\VEC x) = \lambda^i_n(d\VEC u^i|\VEC y^i)\pr{d\VEC x,d\VEC y^i} $, converging to some $\lambda^i_0\in \wp(\ALP U^i\times\ALP Y^i\times\ALP X)$. If Assumption \ref{assm:condind} holds, then for any $m\in\Na$,
\beqq{\lf{n}\int \min\{c,m\} \left(\prod_{i=1}^N \lambda^i_{n}(d\VEC u^i|\VEC y^i)\right)\pr{d\VEC x,d\VEC y^{1:N}}\\
= \int \min\{c,m\} \left(\prod_{i=1}^N \lambda^i_0(d\VEC u^i|\VEC y^i)\right)\pr{d\VEC x,d\VEC y^{1:N}}.}
\end{theorem}
\begin{proof}
Since Assumption \ref{assm:condind} holds, we know from Lemma \ref{lem:control2} that
\beqq{\lambda^i_0(d\VEC u^i,d\VEC y^i,d\VEC x) = \lambda^i_0(d\VEC u^i|\VEC y^i)\pr{d\VEC x,d\VEC y^i}.}
The proof of this theorem then follows from Theorem \ref{thm:intc}.
\end{proof}

This brings us to the following result.
\begin{theorem}\label{thm:sub2}
Under the same hypotheses and notation as in Theorem \ref{thm:sub}, we have
\beqq{\lif{n} \int_{\ALP X\times\ALP Y^{1:N}\times\ALP U^{1:N}} c(\VEC x,\VEC y^{1:N},\VEC u^{1:N})\prod_{i=1}^N\lambda^i_{n}(d\VEC u^i|\VEC y^i)\pr{d\VEC x,d\VEC y^{1:N}}\\
\geq  \int_{\ALP X\times\ALP Y^{1:N}\times\ALP U^{1:N}}c(\VEC x,\VEC y^{1:N},\VEC u^{1:N})\prod_{i=1}^N\lambda^i_{0}(d\VEC u^i|\VEC y^i)\pr{d\VEC x,d\VEC y^{1:N}}.}
\end{theorem}
\begin{proof}
See Appendix \ref{app:sub2}.
\end{proof}

The theorem above says that the expected cost functional of the agents in the team is lower-semicontinuous on the space of behavioral strategies. Since the action sets are non-compact, the set of joint measures over observation and action spaces of each agent is non-compact, and we cannot readily use Weierstrass theorem like results to prove the existence of an optimal solution.

We address this issue in the next subsection. In particular, if the cost function has some stronger (coercivity like) property, then using Assumption \ref{ass:one} and Markov's inequality, we can restrict the search of optimal strategies of the agents to compact sets of joint measures over observation and action spaces of the agents.

\subsection{Compactness of a Set of Probability Measures}
In this subsection, we identify a sufficient condition for a set of measures to be precompact in the weak* topology. We use this result later to show that the search for optimal behavioral strategies of the agents in the team problem can be restricted to a weak* precompact space.

Hereafter, we will use $\ALP A$, $\ALP B$ and $\ALP C$ to denote arbitrary Polish spaces. The following theorem gives a necessary and sufficient condition for a subset of probability measures on a Polish space $\ALP A$ to be weak* precompact.
\begin{theorem}[Prohorov's Theorem]\cite[Theorem 8.6.2, p. 202]{bogachev2006b}
A set $\ALP M\subset\wp(\ALP A)$ is weak* precompact if and only if it is tight, that is, for every $\epsilon>0$, there exists a compact set $\SF K_{\epsilon}\subset \ALP A$ such that $\mu(\ALP A\setminus \SF K_{\epsilon})<\epsilon$ for all $\mu\in\ALP M$.
\end{theorem}

We now define a class of functions and study an important result involving functions in this class.

\begin{definition}[Class $\ic{\ALP A,\ALP B}$]\label{def:icclass}
We say that a non-negative measurable function $\phi:\ALP A\times\ALP B\times\ALP C\rightarrow \Re$ is in class $\ic{\ALP A,\ALP B}$ if $\phi$ satisfies any one of the following two conditions:
\begin{enumerate}
\item For every $M>0$ and for every compact set $\SF K\subset \ALP A$, there exists a compact set $\SF L\subset \ALP B$ such that
\beqq{\inf_{\SF K\times \SF L^\complement\times\ALP C} \phi(\VEC a,\VEC b,\VEC c) \geq M.}
\item For every $M>0$ and every point $\VEC a\in\ALP A$, there exists an open neighborhood $\SF O\subset \ALP A$ of the point $\VEC a\in\ALP A$ and a compact set $\SF L\subset \ALP B$ such that
\beqq{\inf_{\SF O\times \SF L^\complement\times\ALP C} \phi(\VEC a,\VEC b,\VEC c) \geq M.}
\end{enumerate}
We can have $\ALP C=\emptyset$.{\hfill $\Box$}
\end{definition}

A large class of team problems have cost functions that belong to the class of functions defined above, where $\ALP A$ is the space of primitive random variables and $\ALP B$ is an action space of some agent. This class of functions is therefore an important one, and we will exploit this property of cost function to show the existence of an optimal solution in a team. We first identify a few examples of functions in class $\ic{\ALP A,\ALP B}$.

\begin{example}\label{exam:1}
Let $\ALP A=\ALP B = \ALP C=\Re^n$, and define $\phi_1(\VEC a,\VEC b,\VEC c) := \|\VEC b-\VEC a\|+\|\VEC c\|$ and $\phi_2(\VEC a,\VEC b,\VEC c) := \|\VEC b-\VEC a\|^2$. Then, $\phi_1$ and $\phi_2$ are in class $\ic{\ALP A,\ALP B}$. Any non-negative continuous and increasing function on $\Re$ composed with $\phi_1$ or $\phi_2$ is also in class $\ic{\ALP A,\ALP B}$. For example, $\exp(\phi_1(\VEC a,\VEC b,\VEC c))$ and $\exp(\phi_2(\VEC a,\VEC b,\VEC c))$ are in class $\ic{\ALP A,\ALP B}$.
\end{example}

Our next result gives a sufficient condition for a set of measures to be tight, which uses the class of functions introduced in Definition \ref{def:icclass}.
\begin{lemma}[Tightness of a set of Measures]
\label{lem:q2kcom}
Let $\phi:\ALP A\times\ALP B\times\ALP C\rightarrow \Re$ be a non-negative measurable function in the class $\ic{\ALP A,\ALP B}$. Fix $k$ to be a non-negative real number and let $\ALP N\subset\wp(\ALP A)$ be a weak* compact set of measures. Define $\ALP M\subset \wp(\ALP A\times\ALP B\times\ALP C)$ as follows:
\beqq{\ALP M = \bigg\{\mu\in \wp(\ALP A\times\ALP B\times\ALP C):\pj^{\ALP  A}_{\#}\mu \in\ALP N\text{ and }\int \phi\: d\mu\leq k\bigg\}.}
Then, $\pj^{\ALP A\times\ALP B}_{\#}\ALP M$ is a tight set of measures. Furthermore, if $\phi$ is lower semicontinuous, then $\pj^{\ALP A\times\ALP B}_{\#}\ALP M$ is weak* compact.
\end{lemma}
\begin{proof}
See Appendix \ref{app:q2kcom}.
\end{proof}

Now that we have a sufficient condition on when a set of measures is tight, we can look at the original static team problem in the next subsection.

\subsection{Existence of Optimal Strategies}
We need the following assumption on the cost function of the team.
\begin{assumption}\label{ass:u1n}
The cost function $c:\ALP X\times\ALP Y^{1:N}\times \ALP U^{1:N}\rightarrow\Re^+$ is a non-negative continuous function in class $\ic{\ALP X\times\ALP Y^{1:N},\ALP U^{i}}$ for every $i\in[N]$.
{\hfill$\Box$}
\end{assumption}

It should be noted that the conditions in Assumption \ref{ass:u1n} are not dependent on the control strategies that the agents choose. The following lemma identifies a set of tight measures using Lemma \ref{lem:q2kcom}, with the property that any expected cost below a certain threshold is either achieved by measures in that set or cannot be achieved.

\begin{lemma}\label{lem:micompact}
Assume that Team {\bf ST2} satisfies Assumptions \ref{ass:one} and \ref{assm:condind}. Let $\tilde{\pi}^{1:N}\in\ALP R^{1:N}$ be the set of behavioral strategies of the agents which results in finite expected cost to the team. Consider sets $\ALP P^i\subset\ALP R^i, i\in[N]$ such that there exists a set of behavioral strategies $\pi^{1:N}\in\ALP P^{1:N}$ satisfying $J(\pi^{1:N})\leq J(\tilde{\pi}^{1:N})$. Define
\beqq{\ALP M^i:=\Big\{\lambda^i\in\wp(\ALP X\times\ALP Y^i\times\ALP U^i):\lambda^i(d\VEC u^i,d\VEC y^i,d\VEC x)= \pi^i(d\VEC u^i|\VEC y^i)\pr{d\VEC y^i,d\VEC x},\\
\pi^i\in\ALP P^i\Big\},\quad i\in[N]}
If the cost function of the team satisfies Assumption \ref{ass:u1n}, then $\ALP M^i\subset\wp(\ALP X\times\ALP Y^i\times\ALP U^i)$ is a tight set of measures for all $i\in[N]$.
\end{lemma}
\begin{proof}
The statement of the lemma readily follows from Lemma \ref{lem:q2kcom}. Define $\ALP M:=\{\mu\in\wp(\ALP X\times\ALP Y^{1:N}\times\ALP U^{1:N}):\int c\:d\mu\leq J(\tilde{\pi}^{1:N})\}$. For every $i\in[N]$, notice that any $\lambda^i\in\ALP M^i$ satisfies $\lambda^i = \pj^{\ALP X\times\ALP Y^i\times\ALP U^i}_{\#}\mu$ for some $\mu\in\ALP M$. Since $c$ is in class $\ic{\ALP X\times\ALP Y^{1:N},\ALP U^{i}}$, by Lemma \ref{lem:q2kcom}, $\ALP M^i$ is tight.
\end{proof}

Thus, we have identified pre-compact sets of joint measures $\ALP M^i, i\in[N]$ which include the optimal joint measures, if they exist. This brings us to the following main result of the section.
\begin{theorem}\label{thm:stateamindobs}
Assume that the cost function of Team {\bf ST2} satisfies Assumption \ref{ass:u1n}. If Assumptions \ref{ass:one} and \ref{assm:condind} hold, then Team {\bf ST2} admits an optimal solution in deterministic strategies.
\end{theorem}
\begin{proof}
Let $\tilde{\pi}^{1:N}\in\ALP R^{1:N}$ be the set of behavioral strategies  of the agents which results in finite expected cost to the team. From Lemma \ref{lem:micompact}, we know that there exist tight sets of measures $\ALP M^i\subset\wp(\ALP X\times\ALP Y^i\times\ALP U^i), i\in[N]$ that contain the optimal joint measures, if they exist.
Consider a sequence of behavioral strategies $\{\pi^{1:N}_n\}_{n\in\Na}\subset\ALP R^{1:N}$ that satisfies
\beqq{J(\pi^{1:N}_n)\leq J(\tilde{\pi}^{1:N}),\quad \text{ and }\quad \lf{n}J(\pi^{1:N}_n) = \inf_{\pi^{1:N}\in\ALP R^{1:N}} J(\pi^{1:N}).}
Define $\lambda^i_n(d\VEC u^i,d\VEC y^i,d\VEC x) := \pi^i_n(d\VEC u^i|\VEC y^i)\pr{d\VEC y^i,d\VEC x}$ for $i\in[N]$ and $n\in\Na$, and notice that $\{\lambda^i_n\}_{n\in\Na}\subset\ALP M^i$. Since $\{\lambda^i_n\}_{n\in\Na}$ is a tight sequence of measures, we know that there exists a weak* convergent subsequence of measures. For every $i\in[N]$, let $\{\lambda^{i}_{n_k}\}_{k\in\Na}$ be the weak* convergent subsequence of measures converging to $\lambda^{i}_0$. From Lemma \ref{lem:control2}, we know that
\beqq{\lambda^i_0(d\VEC u^i,d\VEC y^i,d\VEC x) = \lambda^i_0(d\VEC u^i|\VEC y^i)\pr{d\VEC y^i,d\VEC x}}
for all $i\in[N]$, which means that the conditional independence property is retained in the limit. Let $\pi^i_0\in\ALP R^i$ be such that $\pi^i_0(d\VEC u^i|\VEC y^i) = \lambda^i_0(d\VEC u^i|\VEC y^i)$. From the result of Theorem \ref{thm:sub2}, we conclude that $\lif{k}J(\pi^{1:N}_{n_k}) \geq J(\pi^{1:N}_0)$. Thus, optimal behavioral strategies of the agents exist, and the optimal behavioral strategy of Agent $i$ is the conditional measure $\lambda^i_0(d\VEC u^i|\VEC y^i)$.

Moreover, applying Blackwell's irrelevant information theorem \cite[p. 457]{yukselbook}, there exists a set of deterministic strategies which achieve the same cost as the one achieved using optimal behavioral strategies of the agents. This completes the proof of the theorem.
\end{proof}

\begin{corollary}\label{cor:stateamindobsdeg}
Consider a two-agent static team {\bf ST2}. Assume that the cost function of Team {\bf ST2} satisfies Assumption \ref{ass:u1n}. If Assumptions \ref{ass:one} and \ref{assm:degraded} hold, then Team {\bf ST2} admits an optimal solution in deterministic strategies.
\end{corollary}
\begin{proof}
The proof follows from arguments similar to those used in the proof of Theorems \ref{thm:degradedstatic} and \ref{thm:stateamindobs}.
\end{proof}

\begin{corollary}\label{cor:stateamindobs}
Assume that the cost function of Team {\bf ST2} is continuous in its arguments and the action spaces of the agents are compact subsets of Polish spaces. Furthermore, assume that Team {\bf ST2} satisfies Assumption \ref{ass:one}. If either Assumption \ref{assm:condind} or Assumption \ref{assm:degraded} holds for Team {\bf ST2}, then the team admits an optimal solution in deterministic strategies.
\end{corollary}
\begin{proof}
If the action spaces of the agents are compact, then the Assumption \ref{ass:u1n} on the cost function holds automatically. Then, we apply the result of Theorem \ref{thm:stateamindobs} to establish the statement.
\end{proof}

In the next section, we use Witsenhausen's static reduction technique to convert a class of dynamic team problems to static teams with independent observations, and then apply the result proved in this section to conclude the existence of optimal strategies in that class of dynamic team problems.

\section{Dynamic Teams}\label{sec:staticreduction}
It was shown in \cite{wit1988} that a large class of $N$-agent $T$-time step dynamic stochastic control problems with certain information structures can be equivalently written as $NT$-agent static optimization problems. In order to define the equivalent static problem, we introduce the following notation:
\beqq{\Omega_0 =\ALP X_1\times\ALP W^0_{1:T}, \qquad \Omega^i_t = \ALP W^i_t, \qquad i\in[N], t\in[T].}
We let $\omega_0$ and $\omega^i_t$ denote generic elements of $\Omega_0$ and $\Omega^i_t$, respectively. Furthermore, we assume that $\Omega_0$ and $\Omega^i_t$ are measure spaces, endowed with the probability measures $\xi_{\Omega_0}$ and $\xi_{\Omega^i_t}$, respectively, which are defined as
\beqq{\xi_{\Omega_0} := \xi_{\ALP X_1}\xi_{\ALP W^0_1}\ldots\xi_{\ALP W^0_T},\qquad \xi_{\Omega^i_t} := \xi_{\ALP W^i_t},\quad i\in[N], t\in[T].}
With this notation, the cost function of the team problem is written as $c:\Omega_0\times\ALP Y^{1:N}_{1:T}\times\ALP U^{1:N}_{1:T}\rightarrow\Re^+$, and we assume that it is continuous. We assume that each agent only observes $\VEC Y^i_t$, that is, its information set is a singleton.


Now, using the static reduction argument, we can transform the original problem to a static team problem with a different cost function. Toward this end, let us rewrite the observations of the agents as
\beq{\label{eqn:obsit}\VEC y^i_t = h^i_t(\omega_0,\omega^i_t, \VEC u^{1:N}_{1:t-1}).}
Note that due to Assumption \ref{ass:stf}, the functions $h^i_t, i\in[N],t\in[T]$ are continuous maps of their arguments. We now make the following assumption.

\begin{assumption}\label{assm:abscont}
For every $(i,t)\in [N]\times[T]$, there exists a probability measure $\nu^i_t\in\wp(\ALP Y^i_t)$ and a continuous function $\varphi^i_t:\ALP Y^i_t\times \Omega_0\times\ALP U^{1:N}_{1:t-1}\rightarrow\Re^+$ such that
\beqq{\pr{\VEC y^i_t\in\SF Y|\omega_0,\VEC u^{1:N}_{1:t-1}}  = \int_{\SF Y} \varphi^i_t(\VEC y^i_t;\omega_0,\VEC u^{1:N}_{1:t-1}) \nu^i_t(d\VEC y^i_t)\quad\text{ for all } \SF Y\in\FLD B(\ALP Y^i_t).}
Define $\varphi:\Omega_0\times\ALP Y^{1:N}_{1:T}\times\ALP U^{1:N}_{1:T}\rightarrow\Re^+$ as
\beqq{\varphi(\omega_0,\VEC y^{1:N}_{1:T},\VEC u^{1:N}_{1:T}):=\prod_{i=1}^N\prod_{t=1}^T\varphi^i_t(\VEC y^i_t;\omega_0,\VEC u^{1:N}_{1:t-1}).}
By definition, $\varphi$ is a continuous function of its arguments.{\hfill $\Box$}
\end{assumption}

In the next lemma, we state a sufficient condition on the mapping $h^i_t$ and the noise statistics $\xi_{\Omega^i_t}$ such that the above assumption is satisfied.
\begin{lemma}\label{lem:gaussiannoise}
Assume that all state, action, observation and noise spaces are Euclidean spaces of appropriate dimensions. For all $i\in[N]$ and $t\in[T]$, let
\beqq{h^i_t(\omega_0,\omega^i_t, \VEC u^{1:N}_{1:t-1}) := \check{h}^i_t(\omega_0, \VEC u^{1:N}_{1:t-1})+\omega^i_t,}
where $\check{h}^i_t$ is a continuous map of its arguments. If $\xi_{\Omega^i_t}$ admits a zero-mean Gaussian density function $\eta^i_t$ for all $i\in[N], t\in[T]$, then Assumption \ref{assm:abscont} holds for the dynamic team problem.
\end{lemma}
\begin{proof}
Note that $\eta^i_t$ is strictly positive at all points in its domain $\Omega^i_t$, and $\Omega^i_t = \ALP Y^i_t$. For every $i\in[N]$ and $t\in[T]$, define $\varphi^i_t$ and $\nu^i_t$ as
\beqq{\varphi^i_t(\VEC y^i_t;\omega_0,\VEC u^{1:N}_{1:t-1}) := \frac{\eta^i_t(\VEC y^i_t-\check{h}^i_t(\omega_0, \VEC u^{1:N}_{1:t-1}))}{\eta^i_t(\VEC y^i_t)},\qquad \nu^i_t(d\VEC y^i_t) = \eta^i_t(\VEC y^i_t)d\VEC y^i_t.}
Since $h^i_t$ is continuous, $\check{h}^i_t$ is a continuous map for all $i\in[N], t\in[T]$. Thus, $\varphi^i_t$ is a continuous map of its arguments. Furthermore, $\varphi^i_t$ is strictly positive in its domain. It is easy to see that with this definition,
 \beqq{\pr{\VEC y^i_t\in\SF Y|\omega_0,\VEC u^{1:N}_{1:t-1}}  = \int_{\SF Y} \varphi^i_t(\VEC y^i_t;\omega_0,\VEC u^{1:N}_{1:t-1}) \nu^i_t(d\VEC y^i_t)\quad\text{ for all } \SF Y\in\FLD B(\ALP Y^i_t),}
which establishes the statement.
\end{proof}

We now define the reduced static team problem corresponding to the dynamic team described above.
\begin{definition}[Reduced static team problem]\label{def:redteam}
Consider the $NT$-agent static team problem with the agents indexed as $(i,t)$. Agent $(i,t)$ observes a random variable $\VEC Y^i_t$ with probability measure $\nu^i_t$, which is independent of observations of all other agents. Agent $(i,t)$, based on the realization $\VEC y^i_t$ of its observation, chooses a control action $\VEC u^i_t$. The cost function for the team is given by
\beqq{c(\omega_0,\VEC y^{1:N}_{1:T},\VEC u^{1:N}_{1:T})\varphi(\omega_0,\VEC y^{1:N}_{1:T},\VEC u^{1:N}_{1:T}).}
We call the static team problem thus defined as reduced static team problem and refer to it as RST problem. {\hfill $\Box$}
\end{definition}

We now recall the following result from \cite{wit1988}, which shows that any dynamic problem and its corresponding reduced static problem are equivalent optimization problems over the same space of strategies of the agents.
\begin{theorem}[\cite{wit1988}]\label{thm:wit1988}
Let $J:\ALP R^{1:N}_{1:T}\rightarrow\Re_+$ be the expected cost functional of the dynamic team problem, and $J_{RST}:\ALP R^{1:N}_{1:T}\rightarrow\Re_+$ be the expected cost functional of the corresponding reduced static team problem, defined as
\beqq{J_{RST}(\pi^{1:N}_{1:T}) = \int c\:\varphi\: \prod_{i=1}^N\prod_{t=1}^T \pi^i_t(d\VEC u^i_t|\VEC y^i_t) \nu^i_t(d\VEC y^i_t)\pr{d\VEC \omega_0},\quad \pi^{1:N}_{1:T}\in \ALP R^{1:N}_{1:T}.}
Then, for any $\pi^{1:N}_{1:T}\in \ALP R^{1:N}_{1:T}$, we have $J(\pi^{1:N}_{1:T})=J_{RST}(\pi^{1:N}_{1:T})$.
\end{theorem}

It should be noted that for any dynamic team problem that admits a reduced static problem, the corresponding RST may not satisfy the hypotheses of Theorem \ref{thm:stateamindobs}. Thus, the results we proved for static teams cannot be applied to conclude the existence of a solution to a dynamic team problem. We illustrate the difficulty in using such a approach in the following example.
  
\subsection*{Witsenhausen's Counterexample}
Witsenhausen's counterexample is a two-agent dynamic LQG team problem, first studied by Witsenhausen in \cite{witcount}. The first agent observes a mean-zero unit variance Gaussian random variable $y_1$ and decides on a real number $U_1$. The second agent observes $Y_2:=U_1+W_2$, where $W_2$ is a mean-zero Gaussian noise with unit variance, and decides on another real number $U_2$. The behavioral strategy space of Agent $i$ is $\ALP R_i$, $i=1,2$. The cost function of the team is given by
\beqq{c_{D}(y_1,u_2,u_2) = (u_1-y_1)^2+(u_2-u_1)^2.}
It is well known that the above dynamic team problem admits an optimal solution \cite{witcount}. The dynamic team problem can be reduced to a static team problem using Lemma \ref{lem:gaussiannoise} \cite{wit1988}. In the corresponding RST problem, each agent observes a mean-zero unit variance Gaussian random variable that is independent of the observation of the other agent. The cost function for the RST is
\beqq{c_{S}(y_1,u_1,u_2) = \Big((u_1-y_1)^2+(u_2-u_1)^2\Big)\exp\left(\frac{-u_1^2+2y_2u_1}{2}\right).}
The cost function for the dynamic problem $c_D$ is in classes $\ic{\ALP Y_1,\ALP U_1}$ and $\ic{\ALP Y_1\times\ALP Y_2\times\ALP U_1,\ALP U_2}$. The cost function $c_S$ for the corresponding RST is in class $\ic{\ALP Y_1\times\ALP Y_2\times\ALP U_1,\ALP U_2}$ (which follows from Lemma \ref{lem:phi12} to be introduced and proved later). However, $c_S$ is not in class $\ic{\ALP Y_1,\ALP U_1}$ because as $|u_1|\rightarrow\infty$, the cost goes to zero for any fixed value of $y_1$. Therefore, the result of Theorem \ref{thm:stateamindobs} is not applicable to the RST problem.

The above example illustrates that the results we obtained for the static team problems in Sections \ref{sec:staticteam} and \ref{sec:unbounded} are not readily applicable to all dynamic team problems that admit static reductions. A certain structure on the cost function of a dynamic team and further assumptions on the corresponding RST problem are needed to prove the existence of a team-optimal solution. In the next subsection, we state the assumptions that we make on the dynamic team problem in order to establish existence.

\subsection{Assumptions on Dynamic Team}
In order to show the existence of optimal strategies in dynamic teams, we assume the following structure.
\begin{assumption}\label{assm:dynteam}
\begin{enumerate}
\item The dynamic team problem satisfies Assumptions \ref{ass:stf}, \ref{ass:one} and \ref{assm:abscont}.
\item The agents in the team do not share their observations with anyone. Any agent who acts more than once does not recall its past observation(s).
\item The cost function $c$ of the dynamic team problem is in the structural form
\beqq{c(\omega_0,\VEC y^{1:N}_{1:T},\VEC u^{1:N}_{1:T}) = \sum_{t=1}^T\sum_{i=1}^N c^i_t(\VEC u^i_t,\omega_0,\VEC u^{1:N}_{1:t-1},\VEC y^{1:N}_{1:t-1},\VEC y^i_t)+\kappa(\omega_0,\VEC y^{1:N}_{1:T},\VEC u^{1:N}_{1:T}),\label{eqn:dyncost}}
where $c^i_t$ is a non-negative and continuous function in the class $\ic{\Omega_0\times\ALP U^{1:N}_{1:t-1}\times\ALP Y^{1:N}_{1:t-1}\times\ALP Y^i_t,\ALP U^i_t}$ for all $i\in[N]$ and $t\in[T]$, and $\kappa$ is a non-negative continuous function of its arguments.
\item For all $i\in[N]$ and $t\in[T]$, the continuous function $\varphi^i_t:\ALP Y^i_t\times \Omega_0\times\ALP U^{1:N}_{1:t-1}\rightarrow\Re^+$, as defined in Assumption \ref{assm:abscont}, is strictly positive at all points in its domain.
\end{enumerate}
\end{assumption}

In the rest of this section, we consider dynamic team problems satisfying Assumption \ref{assm:dynteam}. Let us first recall the following features of the corresponding RST problem:
\begin{enumerate}
\item If the behavioral control strategy of Agent $i$ at time $t$ is $\pi^i_t$, then the joint measure on $\ALP U^i_t\times\ALP Y^i_t$ in the corresponding RST problem is $\pi^i_t(d\VEC u^i_t|\VEC y^i_t)\nu^i_t(d\VEC y^i_t)$.
\item Recall from Assumption \ref{ass:one} that there exists a set of behavioral strategies $\tilde\pi^{1:N}_{1:T}$ that achieves a finite cost $J(\tilde\pi^{1:N}_{1:T})$ in the dynamic team. Since RST and dynamic team problems are equivalent problems (see Theorem \ref{thm:wit1988}), RST also achieves the same cost with the behavioral strategies $\tilde\pi^{1:N}_{1:T}$.
\end{enumerate}

Let $\{\ALP P^i_t\subset\ALP R^i_t\}_{i\in[N],t\in[T]}$ be the set of behavioral strategies of the agents such that there exists $\pi^{1:N}_{1:T}$ satisfying $\pi^i_t\in\ALP P^i_t$ and
\beqq{J(\pi^{1:N}_{1:T})\leq J(\tilde\pi^{1:N}_{1:T}).}
Define $\lambda^i_t(d\VEC u^i_t,d\VEC y^i_t) := \pi^i_t(d\VEC u^i_t|\VEC y^i_t)\nu^i_t(d\VEC y^i_t)$ for $\pi^i_t\in\ALP P^i_t$, and let $\ALP M^i_t$ denote the set of all such $\lambda^i_t$.

If the optimal behavioral strategies of the agents exist, then the optimal behavioral strategy of Agent $(i,t)$ in the RST problem must lie in the set $\ALP M^i_t$. In order to establish the existence of optimal strategies in the dynamic team problem, we show that $\ALP M^i_t$ is a tight set of measures using a similar approach as in Lemma \ref{lem:micompact}. In the next subsection, we prove some auxiliary results that are needed to show that the set of measures $\ALP M^i_t$ is tight. This is a crucial part of the proof of existence of optimal strategies in the dynamic team problem.

\subsection{Auxiliary Results}
Our first auxiliary result is as follows.
\begin{lemma}\label{lem:varphi}
\begin{enumerate}
\item For any $i\in[N]$ and $t\in[T]$ and any $\omega_0\in\Omega_0$ and $\VEC u^{1:N}_{1:t-1}\in\ALP U^{1:N}_{1:t-1}$, we have
\beqq{\int_{\ALP Y^i_t\times\ALP U^i_t} \varphi^i_t(\VEC y^i_t;\omega_0,\VEC u^{1:N}_{1:t-1}) \lambda^i_t(d\VEC u^i_t,d\VEC y^i_t) =1.}
\item For any $i\in[N]$ and $t\in[T]$ and $\lambda^i_t\in\wp(\ALP Y^i_t\times\ALP U^i_t)$,
\beqq{\int_{\Omega_0\times\ALP Y^{1:N}_{1:T}\times\ALP U^{1:N}_{1:T}}  c^i_t(\VEC u^i_t,\omega_0,\VEC u^{1:N}_{1:t-1},\VEC y^{1:N}_{1:t-1},\VEC y^i_t)
\varphi(\omega_0,\VEC y^{1:N}_{1:T},\VEC u^{1:N}_{1:T})\; d\lambda^{1:N}_{1:T}\;\pr{d\omega_0} \\
= \int_{\Omega_0\times\ALP Y^{1:N}_{1:t-1}\times\ALP U^{1:N}_{1:t-1}\times\ALP Y^i_t\times\ALP U^i_t} \bar{c}^i_t(\VEC u^i_t,\omega_0,\VEC u^{1:N}_{1:t-1},\VEC y^{1:N}_{1:t-1},\VEC y^i_t)\;d\lambda^{1:N}_{1:t-1}\;d\lambda^i_t\; \pr{d\omega_0}, }
where
\beq{\bar{c}^i_t(\VEC u^i_t,\omega_0,\VEC u^{1:N}_{1:t-1},\VEC y^{1:N}_{1:t-1},\VEC y^i_t) = c^i_t(\VEC u^i_t,\omega_0,\VEC u^{1:N}_{1:t-1},\VEC y^{1:N}_{1:t-1},\VEC y^i_t)\times\nonumber\\
\varphi^i_t(\VEC y^i_t;\omega_0,\VEC u^{1:N}_{1:t-1}) \times \prod_{s=1}^{t-1}\prod_{j=1}^N \varphi^j_s(\VEC y^j_s;\omega_0,\VEC u^{1:N}_{1:s-1}).\label{eqn:cbarit} \qquad}
\end{enumerate}
\end{lemma}
\begin{proof}
\begin{enumerate}
\item The statement holds for all $i\in[N]$ and $t\in[T]$ by the definition of $\varphi^i_t$ in Assumption \ref{assm:abscont}.
\item This is a consequence of the first statement.
\end{enumerate}
\hspace{0.97\linewidth}
\end{proof}

Recall that we introduced a class of functions $\ic{\cdot,\cdot}$ in Definition \ref{def:icclass}. In the next lemma, we show that if we multiply a function in this class with a lower-semicontinuous function that does not vanish in its domain, then the resulting function also belongs to the same class. We use this result to show that the cost in \eqref{eqn:cbarit} belongs to the class $\ic{\Omega_0\times\ALP U^{1:N}_{1:t-1}\times\ALP Y^{1:N}_{1:t-1}\times\ALP Y^i_t,\ALP U^i_t}$.

\begin{lemma}\label{lem:phi12}
Let $\phi_1:\ALP A\times\ALP B\rightarrow \Re$ be a measurable function and $\phi_2:\ALP A\rightarrow\Re^+$ be a lower-semicontinuous function that is strictly positive everywhere in its domain. If $\phi_1$ is in class $\ic{\ALP A,\ALP B}$, then the product function $\phi:=\phi_1\phi_2$ is also in class $\ic{\ALP A,\ALP B}$.
\end{lemma}
\begin{proof}
Fix $m>0$ and a compact set $\SF K\subset\ALP A$. Define $M := \frac{m}{\min_{\VEC a\in\SF K} \phi_2(\VEC a)}$. Then, by the property of class $\ic{\ALP A,\ALP B}$ functions, there exists a compact set $\SF L\subset\ALP B$, depending on $\SF K$ and $M$, such that
\beqq{\inf_{\SF K\times\SF L^\complement} \phi_1(\VEC a,\VEC b)\geq M.}
Now, due to the property of infimum, we get
\beqq{\inf_{\SF K\times\SF L^\complement} \phi(\VEC a,\VEC b)\geq \inf_{\SF K\times\SF L^\complement} \phi_1(\VEC a,\VEC b)\min_{\SF K} \phi_2(\VEC a) = m,} which completes the proof of the statement.
\end{proof}

As a result of the lemma above, we have the following fact.
\begin{lemma}\label{lem:cbarit}
The function $\bar{c}^i_t$, as defined in \eqref{eqn:cbarit}, is a non-negative and continuous function in the class $\ic{\Omega_0\times\ALP U^{1:N}_{1:t-1}\times\ALP Y^{1:N}_{1:t-1}\times\ALP Y^i_t,\ALP U^i_t}$.
\end{lemma}
\begin{proof}
By Assumption \ref{assm:dynteam}, $c^i_t(\VEC u^i_t,\omega_0,\VEC u^{1:N}_{1:t-1},\VEC y^{1:N}_{1:t-1},\VEC y^i_t)$ is in the class $\ic{\Omega_0\times\ALP U^{1:N}_{1:t-1}\times\ALP Y^{1:N}_{1:t-1}\times\ALP Y^i_t,\ALP U^i_t}$. Now, in the statement of Lemma \ref{lem:phi12}, take the functions $\phi_1$ and $\phi_2$ as
\beqq{\phi_1(\VEC u^i_t,\omega_0,\VEC u^{1:N}_{1:t-1},\VEC y^{1:N}_{1:t-1},\VEC y^i_t) &:=& c^i_t(\VEC u^i_t,\omega_0,\VEC u^{1:N}_{1:t-1},\VEC y^{1:N}_{1:t-1},\VEC y^i_t)\\
\phi_2(\omega_0,\VEC u^{1:N}_{1:t-1},\VEC y^{1:N}_{1:t-1},\VEC y^i_t) &:=& \varphi^i_t(\VEC y^i_t;\omega_0,\VEC u^{1:N}_{1:t-1}) \times \prod_{s=1}^{t-1}\prod_{j=1}^N \varphi^j_s(\VEC y^j_s;\omega_0,\VEC u^{1:N}_{1:s-1}),}
and note that $\phi_2$ is a continuous and strictly positive function in its domain (see Part 4 of Assumption \ref{assm:dynteam}). As a consequence of the result in Lemma \ref{lem:phi12}, we obtain that the function $\bar{c}^i_t$ is non-negative and continuous in the class $\ic{\Omega_0\times\ALP U^{1:N}_{1:t-1}\times\ALP Y^{1:N}_{1:t-1}\times\ALP Y^i_t,\ALP U^i_t}$. This completes the proof of the lemma.
\end{proof}

We now have all the auxiliary results needed for showing existence. In the next subsection, we prove that any dynamic team as described above admits a team-optimal solution in deterministic strategies of the agents.

\subsection{Proof of Existence of Optimal Strategies}
Our first result in this subsection is that $\ALP M^i_t$ is a tight set of measures for all $i\in[N]$ and $t\in[T]$. This implies that during the search of optimal strategies of the agents in the RST problem, we can restrict the joint measures on the action and observation spaces of Agent $(i,t)$ to a tight set of measures $\ALP M^i_t$.
\begin{lemma}\label{lem:mittight}
The set of measures $\ALP M^i_t$ is tight for all $i\in[N]$ and $t\in[T]$.
\end{lemma}
\begin{proof}
See Appendix \ref{app:mittight}.
\end{proof}

We now turn our attention to showing the existence of optimal strategies in dynamic team problems. We use the result of the lemma above to prove this fact in the next theorem.
\begin{theorem}\label{thm:dynteam}
If a dynamic team problem satisfies Assumption \ref{assm:dynteam}, then it admits a team-optimal solution in deterministic strategies.
\end{theorem}
\begin{proof}
Consider a sequence of behavioral strategies of the agents $\{(\pi^{1:N}_{1:T})_n\}_{n\in\Na}\subset\ALP M^{1:N}_{1:T}$ that satisfies $\lf{n}J((\pi^{1:N}_{1:T})_n) = \inf J(\pi^{1:N}_{1:T})$, where $J$ is the expected cost functional of the dynamic team problem. Let $\{(\lambda^i_t)_n\}_{n\in\Na}\subset\ALP M^i_t$ be defined as
\beqq{(\lambda^i_t)_n(d\VEC u^i_t,d\VEC y^i_t) = (\pi^i_t)_n(d\VEC u^i_t|\VEC y^i_t) \nu^i_t(d\VEC y^i_t),\qquad n\in\Na.}
Since $\{(\lambda^i_t)_n\}_{n\in\Na}$ is a tight sequence of measures, we know that there exists a weak* convergent subsequence of measures. For every $i\in[N]$ and $t\in[T]$, let $\{(\lambda^i_t)_{n_k}\}_{k\in\Na}$ be the weak* convergent subsequence of measures converging to $(\lambda^i_t)_0$. Define the behavioral strategy $(\pi^i_t)_0(d\VEC u^i_t|\VEC y^i_t) = (\lambda^i_t)_0(d\VEC u^i_t|\VEC y^i_t)$ for all $i\in[N]$ and $t\in [T]$. Since $c\varphi$ is a continuous function, from the result of Theorem \ref{thm:sub2}, we conclude that
\beqq{\lif{k}\int c\;\varphi\; d(\lambda^{1:N}_{1:T})_{n_k}\pr{d\omega_0}\geq \int c\;\varphi\; d(\lambda^{1:N}_{1:T})_0\pr{d\omega_0},}
or, equivalently, $\lif{k}J((\pi^{1:N}_{1:T})_{n_k}) \geq  J((\pi^{1:N}_{1:T})_0)$. Thus, optimal behavioral strategies of the agents exist in the RST team problem, and the optimal behavioral strategy of Agent $(i,t)$ is the conditional measures $(\pi^i_t)_0(d\VEC u^i|\VEC y^i)$. Since RST is equivalent to a dynamic team problem, this is also the optimal behavioral strategy of Agent $i$ at time $t$ in the dynamic team problem.

Moreover, by Blackwell's irrelevant information theorem \cite[p. 457]{yukselbook}, there exists a set of deterministic strategies which achieves the same cost as the one achieved using optimal behavioral strategies of the agents, which establishes the result.
\end{proof}

We have the following corollary to the theorem above for dynamic teams in which the agents have compact action spaces. We do not require Parts 3 and 4 of  Assumption \ref{assm:dynteam} to show the existence of optimal strategies for this problem.

\begin{corollary}\label{cor:dynteamext}
Let us consider a dynamic team problem in which $\ALP U^i_t$ is compact for all $i\in[N]$ and $t\in[T]$. If Parts 1 and 2 in Assumption  \ref{assm:dynteam} hold, then the dynamic team problem admits an optimal solution in deterministic strategies of the agents.
\end{corollary}
\begin{proof}
As a result of Assumption \ref{assm:abscont}, the dynamic team problem is equivalent to a reduced static team problem defined in Definition \ref{def:redteam}. Note that due to the assumption, the cost function of the reduced static team problem is continuous. Applying the result of Corollary \ref{cor:stateamindobs} to the reduced static team problem, we conclude that the reduced problem admits an optimal solution in deterministic strategies. The optimal strategy of Agent $(i,t)$ in the reduced static team problem is also the optimal strategy of Agent $i$ at time step $t$ in the dynamic team problem due to the equivalence of the two team problems. This completes the proof of the theorem.
\end{proof}

We now revisit Witsenhausen's counterexample.
\subsection*{Revisiting Witsenhausen's Counterexample}
Recall that we were unable to prove the existence of a solution to the corresponding RST problem of Witsenhausen's counterexample using the results we obtained for static team problems in Sections \ref{sec:staticteam} and \ref{sec:unbounded}. We now outline the essential steps of the proof above adapted to the RST of Witsenhausen's counterexample. 

Let $\nu_1\in\wp(\ALP Y_1)$ and $\nu_2\in\wp(\ALP Y_2)$ be probability measures that admit mean-zero unit-variance Gaussian density functions. First note that if both agents apply zero control, then the expected cost is $\ex{Y_1^2} = 1$, which is finite. Let $\ALP R_1$ and $\ALP R_2$ be the behavioral strategy spaces of the first and the second controller, respectively. Let $\ALP P_1\subset\ALP R_1$ and $\ALP P_2\subset\ALP R_2$ be the sets of behavioral strategies of the controllers such that there exist $\pi_1\in\ALP P_1$ and $\pi_2\in\ALP P_2$ which yield $J(\pi_1,\pi_2)\leq 1$. Now, the following four steps lead to the existence of a team-optimal solution to this problem with non-classical information.   

\begin{enumerate}
\item For any $\pi_2\in\ALP R_2$, we have
\beqq{ & & \int_{\ALP Y_2\times\ALP U_2} c_S\: \pi_2(du_2|y_2)\nu_2(dy_2)\\
&\geq& \int_{\ALP Y_2\times\ALP U_2} (u_1-y_1)^2\exp\left(\frac{-u_1^2+2y_2u_1}{2}\right)\frac{1}{\sqrt{2\pi}}\exp\left(-\frac{y_2^2}{2}\right) \pi_2(du_2|y_2)dy_2,\\
 &=& (u_1-y_1)^2\int_{\ALP Y_2\times\ALP U_2} \frac{1}{\sqrt{2\pi}}\exp\left(\frac{-(y_2-u_1)^2}{2}\right)\pi_2(du_2|y_2)dy_2\\
& = & (u_1-y_1)^2, }
where the first inequality follows from dropping the quadratic term $(u_2-u_1)^2$ from the expression of $c_S$, the second equality is immediate, and the third equality follows from the fact that
\beqq{\pi_2(du_2|y_2) \:\frac{1}{\sqrt{2\pi}}\exp\left(\frac{-(y_2-u_1)^2}{2}\right)dy_2}
is a probability measure over $\ALP U_2\times\ALP Y_2$.
This is also a consequence of Lemma \ref{lem:varphi}.
\item The function $(u_1-y_1)^2$ is in class $\ic{\ALP Y_1,\ALP U_1}$ and $(u_2-u_1)^2\exp\left(\frac{-u_1^2+2y_2u_1}{2}\right)$ is in class $\ic{\ALP U_1\times\ALP Y_1\times\ALP Y_2,\ALP U_2}$ by Lemma \ref{lem:cbarit}.
\item The set of measures $\ALP M_i$, defined by 
\beqq{\ALP M_i:=\Big\{\lambda_i\in\wp(\ALP Y_i\times\ALP U_i):\lambda_i(du_i,dy_i) = \pi_i(du_i|y_i)\nu_i(dy_i), \:\pi_i\in\ALP P_i\Big\}}
is tight for $i\in\{1,2\}$ by Lemma \ref{lem:mittight}. The proof essentially uses Points 1 and 2 above, coupled with Lemma \ref{lem:q2kcom} in a sequential fashion. Using Point 1, we conclude
\beqq{\int_{\ALP U_1\times\ALP Y_1} (u_1-y_1)^2 \lambda_1(du_1,dy_1) &\leq&  \ex{Y_1^2} =  1} 
for all $\lambda_1\in\ALP M_1$ (or equivalently $\pi_1\in\ALP P_1$). Then, using Point 2 and Lemma \ref{lem:q2kcom}, we conclude that the set of measures $\ALP M_1$ is tight. Now, notice that
\beqq{\int_{\ALP Y_{1:2}\times\ALP U_{1:2}}(u_2-u_1)^2\exp\left(\frac{-u_1^2+2y_2u_1}{2}\right)\:d\lambda_1\:d\lambda_2 &\leq&   \ex{Y_1^2} = 1.}
Since $\ALP M_1$ is tight and $(u_2-u_1)^2\exp\left(\frac{-u_1^2+2y_2u_1}{2}\right)$ is in class $\ic{\ALP U_1\times\ALP Y_1\times\ALP Y_2,\ALP U_2}$ (see Point 2 above), we conclude that $\ALP M_2$ is tight by Lemma \ref{lem:q2kcom}.
\item Finally, using same arguments as in the proof of Theorem \ref{thm:stateamindobs}, one can conclude that there exist optimal strategies of the agents.
\end{enumerate}
Notice that the above proof of existence of a solution to Witsenhausen's counterexample is completely different from either of the proofs given in \cite{witcount} and \cite{wu2011}.

This concludes the discussion in this section. In the next section, we show the existence of optimal solution in LQG team problems with ``no observation sharing'' information structures using the results of this section.

\section{LQG Teams}\label{sec:lqgteam}
We now consider a class of dynamic team problems in which the state, action and observation spaces are Euclidean spaces, the state transition and observation functions are linear, and the primitive random variables are mutually independent Gaussian random variables. In particular, we assume that the observation equation for Agent $(i,t)$ is given by
\beq{\label{eqn:obsitlqg}\VEC y^i_t = h^i_t(\omega_0,\VEC u^{1:N}_{1:t-1})+\omega^i_t,}
where $h^i_t$  is a linear function of its arguments and $\omega^i_t$ is a zero-mean Gaussian random vector with positive definite covariance.

We assume that the cost function of the dynamic team problem is quadratic in the actions of the agents and is of the following form:
\beq{c(\omega_0,\VEC y^{1:N}_{1:T},\VEC u^{1:N}_{1:T}) &=& \sum_{t=1}^T\sum_{i=1}^N \|\VEC u^i_t - p^i_t(\omega_0,\VEC u^{1:N}_{1:t-1},\VEC y^{1:N}_{1:t-1},\VEC y^i_t)\|_{R^i_t}^2\nonumber\\ & & +\kappa(\omega_0,\VEC y^{1:N}_{1:T},\VEC u^{1:N}_{1:T}),\label{eqn:lqgcost} }
where $\{R^i_t\}_{i\in[N],t\in[T]}$ is a sequence of positive definite matrices of appropriate dimensions, $\{p^i_t:\Omega_0\times\ALP U^{1:N}_{1:t-1}\times\ALP Y^{1:N}_{1:t-1}\times\ALP Y^i_t\rightarrow\ALP U^i_t\}_{i\in[N],t\in[T]}$ is a sequence of continuous functions\footnote{In most cases of interest, $\{p^i_t\}_{i\in[N],t\in[T]}$ are linear maps, which is the reason why we have called this class of teams LQG, realizing that in general, with $p^i_t$'s nonlinear, $c$ is not going to be quadratic in the $u^i_t$'s.} and $\kappa$ is a non-negative continuous function. We henceforth refer to teams satisfying the above assumptions and having a cost function of the form \eqref{eqn:lqgcost} as LQG team problems with no observation sharing, and address existence of team-optimal solutions below.

We now turn our attention to showing the existence of optimal strategies in LQG team problems with no observation sharing.
\begin{theorem}\label{thm:lqgteam}
Consider a dynamic LQG team problem as formulated above, where the agents do not share their observations and the observation of each agent as given by \eqref{eqn:obsitlqg} is corrupted by additive Gaussian noise. If the cost is given by \eqref{eqn:lqgcost}, then the dynamic LQG team admits a team-optimal solution in deterministic strategies.
\end{theorem}
\begin{proof}
In order to establish the result, we need to verify that all parts of Assumption \ref{assm:dynteam} are satisfied by the LQG team problem.

The linearity of state transition and observation equations implies that Assumption \ref{ass:stf} is satisfied and $\{h^i_t\}_{i\in[N],t\in[T]}$, as defined in \eqref{eqn:obsitlqg}, are continuous functions. If we apply zero control action, then the expected cost is finite because the cost is quadratic in the primitive random variables and their distributions are Gaussian. Thus, Assumption \ref{ass:one} is satisfied. Since the observation noises are additive and Gaussian, Assumption \ref{assm:abscont} is satisfied. Furthermore, due to the Gaussian nature of observation noise, we also conclude that $\varphi^i_t$ is strictly positive at all points in its domain for all $i\in[N]$ and $t\in[T]$ (see the proof of Lemma \ref{lem:gaussiannoise}).

The cost function $c$ is continuous. Since $p^i_t$ is continuous, the function $\|\VEC u^i_t - p^i_t(\omega_0,\VEC u^{1:N}_{1:t-1},\VEC y^{1:N}_{1:t-1},\VEC y^i_t)\|_{R^i_t}^2$ lies in the class $\ic{\Omega_0\times\ALP U^{1:N}_{1:t-1}\times\ALP Y^{1:N}_{1:t-1}\times\ALP Y^i_t,\ALP U^i_t}$ for all $i\in[N]$ and $t\in[T]$. The statement is then simply a consequence of Theorem \ref{thm:dynteam}.
\end{proof}

We have thus identified conditions under which an LQG team problem admits an optimal solution. In the next section, we consider a number of well-studied LQG team problems from the literature and establish the existence of team-optimal strategies. Existence of optimal strategies in some of the team problems formulated in the next section is established here for the first time.

\section{Examples}\label{sec:examples}
In this section, we present some examples of LQG teams with the ``no observation sharing'' information structure. In all the examples, Theorem \ref{thm:lqgteam} leads to the conclusion that team-optimal strategies exist. Except for scalar Witsenhausen's counterexample and the Gaussian test channel, existence of optimal strategies was not known for any of the LQG teams considered in this section.

\subsection{One-Agent Finite Horizon (Static Output Feedback) LQG problem}
Consider a linear system in which all primitive random variables are Gaussian and mutually independent of each other. The agent has a stagewise additive quadratic cost function. The information available to the controller at time $t$ is $\VEC Y_t$, where $\VEC Y_t = H_t\VEC X_t+\VEC W_t$ for some matrix $H_t$ of appropriate dimensions, that is, we have a static output feedback problem. The total cost to the controller is 
\beqq{c(\VEC x_{1:T+1}, \VEC u_{1:T}) = \sum_{t=1}^T \left( \VEC x^\transpose_{t+1}Q\VEC x_{t+1}+\VEC u^\transpose_{t}R\VEC u_t \right), \qquad Q\geq 0, R>0.}
Since this is an LQG problem with no sharing of observation, it satisfies both hypotheses of Theorem \ref{thm:lqgteam}. Using Theorem \ref{thm:lqgteam}, we then conclude that an optimal static output feedback controller exists. This solution, however, need not be linear \cite{basar1976}.

\subsection{The Gaussian Test Channel}
The Gaussian test channel consists of an encoder and a decoder. The source observes a zero mean Gaussian random variable $X_1$ with variance $\sigma_1^2$, which is encoded by the encoder (Agent 1), and the encoded symbol $U_1$ is sent across a noisy channel to a decoder. The additive noise $W_2$ on the channel is assumed to be a zero mean Gaussian random variable with variance $\sigma_w^2$. The decoder (agent 2) observes the corrupted message $Y_2$, and estimates the realization of the random variable $X_1$ available at the source. The decoder's estimate is denoted by $U_2$.

\begin{figure}[bth]
\centering
\scalebox{0.7}{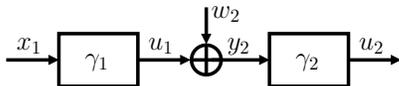}
\caption{\label{fig:wit}A figure depicting the unified setup of the Gaussian test channel and Witsenhausen's counterexample from \cite{bansal1987}.}
\end{figure}

The information structure of the encoder is $ I^1_1 =\{X_1\}$ and of the decoder is $I^2_2 = \{Y_2\}$. Thus, this is an example of a team with asymmetric information. Note that the observation of Agent 2 satisfies the observation equation \eqref{eqn:obsitlqg}. Lemma \ref{lem:gaussiannoise} implies that this dynamic team problem admits a static reduction. The cost function of the team of encoder and decoder is\footnote{For the Gaussian test channel, this corresponds to ``soft-constrained'' version; the standard version has a second moment (hard) constraint on $u_1$. One can show, however, that existence of an optimal solution to one implies existence to the other, and {\it vice versa}.}
\beqq{c(x_1,u_1,u_2) = \lambda u_1^2+(u_2-x_1)^2, \qquad \lambda>0.}
One can check that the cost function of the team is of the form in \eqref{eqn:lqgcost}. It is well known that the optimal encoding and decoding strategies are linear in their arguments, despite the fact that the information structure is non-nested. The only known proof of this result (and therefore of the existence of a solution to this team problem) is an indirect one, that uses information theoretic concepts; see, for example \cite{bansal1987}. We now have here another proof of the existence of team-optimal strategies to the Gaussian test channel as a consequence of Theorem \ref{thm:lqgteam}.

\begin{remark}
The existence result also holds for the more general two-agent LQG problem introduced in \cite{bansal1987}, which subsumes the Gaussian test channel and Witsenhausen's counterexample as special cases. For such extensions, see also \cite{basar2008}.
\end{remark}

\subsection{Multidimensional Gaussian Test Channel and Witsenhausen's Co- unterexample}
Consider the setup depicted in Figure \ref{fig:wit}, with a difference that all random vectors take values in finite dimensional Euclidean spaces of appropriate dimensions. Furthermore, we assume that $\VEC W_2$ has a strictly positive definite covariance, and the entries in $\VEC W_2$ can be correlated. Consider the cost function of the team as
\beqq{c(\VEC x_1,\VEC u_1,\VEC u_2) = \lambda \|\VEC u_1\|^2+\|\VEC u_2-H \VEC x_1\|^2, \qquad \lambda>0,}
where $H$ is a matrix of appropriate dimensions.

It has been shown that under some specific assumptions on the covariance matrix of the noise variable $\VEC W_2$, optimal encoding and decoding schemes exist in the multidimensional Gaussian test channel, again using information theoretic tools; see \cite[Section 11.2.3]{yukselbook} and references therein for a review of such results.  In particular, if certain ``matching conditions'' hold, that is, if the rate distortion achieving transition kernel is matched with the channel capacity achieving source distribution (see \cite{gastpar2008uncoded} and Remark 11.2.1 in \cite{yukselbook} in the context of Gaussian systems), then optimal encoding-decoding strategies will exist.

As in the scalar case, the multidimensional Gaussian test channel admits static reduction and the cost function has the same form as in \eqref{eqn:lqgcost}. Theorem \ref{thm:lqgteam} implies that optimal encoding-decoding strategies exist even if $\VEC X_1$, $\VEC U_1$ and $\VEC U_2$ take values in different Euclidean spaces. Thus, a large class of multidimensional Gaussian test channel problems admits optimal solutions.

A vector version of Witsenhausen's counterexample has also been studied recently \cite{grover2010}. In this formulation, $\VEC X_1,\VEC U_1,\VEC W_1$ and $\VEC U_2$ are all vectors in $\Re^n$ with primitive random variables $\VEC X_1$ and $\VEC W_1$ being mutually independent Gaussian random vectors. Until now, it was not known if vector versions of Witsenhausen's counterexample admit optimal solutions. We now know that the answer is in the affirmative, thanks to Theorem \ref{thm:lqgteam}.

\subsection{A Gaussian Relay Channel}\label{sub:grc}
Consider now the Gaussian relay channel depicted in Figure \ref{fig:GCE}. It comprises an encoder, a certain number of relays and a decoder. The encoder encodes its observation and transmits it over the communication channel. The first relay receives the transmitted signal with an additive noise, re-encodes it and transmits it to the next relay. Thereafter, each relay observes the signal that is transmitted by the previous relay with an additive noise, re-encodes it and transmits it to the next relay. The decoder receives an additive noise corrupted signal transmitted by the last relay, and then decodes it to obtain the best possible estimate of the encoder's observation in the mean-square sense. All the primitive random variables are assumed to be mutually independent and have Gaussian distributions.

This problem was formulated in \cite{lipsa2011} where the authors have shown that non-linear strategies outperform linear strategies when there are two or more relays. Zaidi et al. studied this problem in \cite{zaidi2011}, and they showed that in fact, even with one relay, quantization based strategies outperform linear strategies of the agents. Thus, linear strategies of encoder and decoder are optimal only in the case of the Gaussian test channel discussed earlier, but not in the case of the Gaussian relay channel.

\begin{figure}[bth]
\centering
\scalebox{0.7}{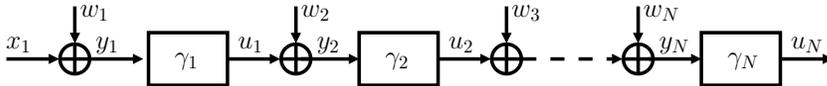}
\caption{\label{fig:GCE}A figure depicting the arrangement of encoder, relays, and decoder in the Gaussian relay channel. Agent 1 is the encoder, Agent $N$ is the decoder while Agents 2 to $N-1$ are relays in the figure.}
\end{figure}

For a concrete formulation of the problem, assume that there are $N-2$ relays and all state, action, noise and observation spaces are the real line. The encoder observes a noise corrupted version of zero-mean Gaussian random variable, $X_1$, with variance $\sigma_1^2$. The observation noise of the encoder denoted by $W_1$ and the observation of the encoder is $Y_1:=X_1+W_1$. The encoder's information is $ I^1 = \{Y_1\}$ and the action of the encoder is $U_1$. For $i\in\{2,\ldots,N-1\}$, the $i^{th}$ relay receives a noise corrupted version of the transmitted signal, denoted by $Y_i := U_{i-1}+W_i$, and $ I^i =\{Y_i\}$. The $i^{th}$ relay outputs $U_i$. Finally, the decoder receives $Y_N$ and outputs $U_N$, which is an estimate of the realization of the random variable $X_1$. The noise variables $W_i,i\in [N]$ are assumed to be pairwise independent, mean-zero Gaussian random variables with some specified variances and independent of the random variable $X_1$. Since the observations of the Agents $2,\ldots,N-
1$ satisfy the observation equation \eqref{eqn:obsitlqg}, we conclude that the dynamic team problem admits a static reduction using Lemma \ref{lem:gaussiannoise}. The cost function $c$ of the team is
\beqq{c(x_1, u_{1:N}) = (u_{N}-x_1)^2+\sum_{i=1}^{N-1}\lambda_iu_i^2,\qquad \lambda_i>0.}
As we mentioned earlier, it is known for this problem that for any number of relays, non-linear strategies outperform best linear ones. However, it is not known whether there exist optimal strategies for the agents.

Since the problem admits a static reduction and the cost function of the team is of the form in \eqref{eqn:lqgcost}, we conclude that optimal encoding, decoding and relay strategies exist for this problem by the result of Theorem \ref{thm:lqgteam}. It is not difficult to see that the same line of reasoning (along with Theorem \ref{thm:lqgteam}) applies to prove that optimal strategies exist for agents in a vector version of this problem as well, where all random variables take values in appropriate dimensional Euclidean spaces (not necessarily of the same dimensions).

\section{Conclusion}\label{sec:concteam}
In this paper, we have identified a set of sufficient conditions on a stochastic team problem with ``no observation sharing'' information structure that guarantee an optimal solution to the team problem. In particular, in a static team problem, if the cost function is continuous and has a certain structure, and the observation channels satisfy certain technical conditions, then there exists an optimal solution.

We used Witsenhausen's static reduction technique to obtain a similar result for a class of dynamic team problems with no sharing of observations among the agents. Furthermore, we proved that LQG team problems with no sharing of observations admit team-optimal solutions under some technical conditions on the cost functions. As a consequence of one of the main results of the paper, we also showed that several dynamic LQG team problems from the literature admit team-optimal solutions. The approach developed in the paper and the specific results obtained settle for good a number of open questions on the existence of optimal strategies in dynamic teams. Furthermore, the results of this paper can be applied to optimal real-time coding/decoding problems to prove the existence of optimal strategies, where it has not been known earlier if optimal policies exist and the few special cases where such results exist have relied on strict information theoretic source-channel matching conditions.

As for the future work, one goal is to obtain approximately optimal policies and numerical techniques to make such optimization problems tractable. Another goal is to obtain explicit analytical solutions for some classes of dynamic teams. Obtaining conditions under which a team-optimal solution exists in a team with ``observation sharing'' information structure is also an important area for further research. Finally, informational aspects of non-cooperative stochastic games is a further relevant area of study.

\section*{Acknowledgement}
The authors gratefully acknowledge incisive comments from Yihong Wu on an initial draft, which has improved several parts of the paper.

\appendix

\section{A Team with Observation Sharing: A Counterexample}\label{app:counterexample}
In this appendix, we consider a two-agent static team problem in which the observation of Agent 2 is shared with Agent 1. The purpose of this counterexample is to illustrate that weak* topology on the space of joint measures of agents' actions and information is not sufficient for the existence of optimal strategies in teams with observation sharing information structures. 

Let $\ALP Y^1 = \ALP Y^2 = [0,1]$ and $\ALP U^1=\ALP U^2 = \{0,1\}$. For $i\in\{1,2\}$, let $Y^i$ be a uniformly distributed random variable taking values in $\ALP Y^i$ that is observed by Agent $i$. Further, assume that $Y^1$ and $Y^2$ are mutually independent random variables. Take the cost function of the team as
\beqq{c(y^{1:2},u^{1:2}) = u^1(1-u^2).}
Note that the cost is a non-negative, continuous, and bounded function of its arguments. Agent 1 decides on $u^1$ based on the realizations $y^1$ and $y^2$, and Agent 2 decides on $u^2$ based on the realization $y^2$. We show that if we take weak* convergent sequences of measures $\{\pi^1_n(du^1|y^1,y^2)dy^1\:dy^2\}_{n\in\Na}$ and $\{\pi^2_n(du^2|y^2)dy^2\}_{n\in\Na}$ that preserve informational constraints in the limit, then the corresponding sequence of joint measures over observations and actions of both agents, that is,
\beqq{\Big\{\pi^1_n(du^1|y^1,y^2)\pi^2_n(du^2|y^2)dy^1\:dy^2\Big\}_{n\in\Na}}
may not converge in the weak* limit. Consequently, a result similar to that of Theorem \ref{thm:intc} may not hold for static teams with observation sharing information structures. This also shows that we need stronger assumptions on the underlying distributions of the primitive random variables and topologies on the measure spaces to show the existence of optimal strategies in teams with observation sharing information patterns.

We now construct the sequences $\{\pi^1_n\}_{n\in\Na}$ and $\{\pi^2_n\}_{n\in\Na}$. For any $n\in\Na$, define $h_n:[0,1]\rightarrow\{0,1\}$ as
\[ h_n(y) = \left\{\begin{array}{cl}
1 & \textrm{if } [2^n y] \textrm{ is even}\\
0 & \textrm{otherwise}.\end{array}\right.\]
Define $A_n\subset [0,1]$ as
\begin{eqnarray}
A_n = \{1\}\cup \bigcup_{k=0}^{2^{n-1}-1} \Big[\frac{2k}{2^n}, \frac{2k+1}{2^n}\Big).\label{eqn:an}
\end{eqnarray}
Then, $h_n(y)=1$ for all $y\in A_n$ and $0$ otherwise. Note that Lebesgue measures of $A_n$ and $A_n^\complement$ are equal and $\frac{1}{2}$. Consider sequences of strategies of Agents $1$ and $2$, given by
\[\pi^1_{n}(du^1|y^1,y^2) = \ind{h_n(y^1) h_n(y^2)}(du^1),\qquad \pi^2_{n}(du^2| y^2) = \ind{h_n(y^2)}(du^2),\qquad n\in\Na.\]
We have the following result, the proof of which is omitted.
\begin{lemma}\label{lem:pi1pi2}
The sequences of joint measures $\{\pi^1_{n}(du^1|y^1,y^2)dy^1dy^2\}_{n\in\Na}$ and $\{\pi^2_{n}(du^2|y^2)dy^2\}_{n\in\Na}$ converge in the weak* sense, respectively, to
\beqq{\pi^1_0(du^1|y^1,y^2)dy^1dy^2 &=& \left( \frac{3}{4}\ind{0}(du^1)+\frac{1}{4}\ind{1}(du^1) \right)dy^1dy^2,\\
\pi^2_0(du^2|y^2)dy^2 &=& \left(\frac{1}{2}\ind{0}(du^2)+\frac{1}{2}\ind{1}(du^2)\right)dy^2.}
\end{lemma}

Note that $u^1$ and $u^2$ are independent of the realizations $y^1$ and $y^2$ in the limit. For any natural number $n$, we have
\beqq{& & \int c(y^1,y^2,u^1,u^2)\pi^1_{n}(du^1|y^1,y^2)\pi^2_{n}(du^2| y^2) dy^1dy^2 \\
 & & = \int c(y^1,y^2,h_n(y^1)h_n(y^2),h_n(y^2)) dy^1dy^2,\\
& & =h_n(y^1)h_n(y^2)(1-h_n(y^2)) = 0.}
However, in the limit
\beqq{\int c(y^1,y^2,u^1,u^2)\pi^1_0(du^1|y^1,y^2)\pi^2_0(du^2| y^2) dy^1dy^2 = \pr{u^1 = 1, u^2=0} = \frac{1}{8}.}
Hence, the sequence of measures $\Big\{\pi^1_{n}(du^1|y^1,y^2)\pi^2_{n}(du^2|y^2)dy^1\:dy^2\Big\}_{n\in\Na}$ does not converge to $\pi^1_0(du^1|y^1,y^2)\pi^2_0(du^2|y^2)dy^1\:dy^2$ in the weak* topology.


Our analysis and existence result built upon the properties of weak* convergent sequence of measures. It is evident from the above counterexample that this convergence notion is not sufficient, and a stronger notion of topologies on the measure spaces is needed to show the existence of solutions to teams with observation sharing information structures\footnote{Certain instances of such information structures can be viewed as centralized information structures \cite{NayyarBookChapter}, albeit over a much larger state and action spaces of the agents.}. We leave this topic for future research.

\section{Proof of Lemma \ref{lem:unifequi}}\label{app:unifequi}
Since $g$ is uniformly continuous, we can assume that
for every $\epsilon>0$, there exists a $\delta>0$ such that for all $\VEC a_0\in\ALP A$, $\VEC c_0\in\ALP C$ and $\VEC a\in\ALP A$, $\VEC c\in\ALP C$ satisfying $d_{\ALP A}(\VEC a,\VEC a_0)<\delta$, $d_{\ALP C}(\VEC c,\VEC c_0)<\delta$, we have
\beqq{\sup_{\VEC b\in\ALP B}|g(\VEC a,\VEC b,\VEC c)-g(\VEC a_0,\VEC b,\VEC c_0)|<2\epsilon.}
Let $M:=\sup_{\VEC a\in\ALP A} \int h_{(\eta,\nu)}d\nu<\infty$. Consider any probability measure $\mu\in\wp(\ALP Y\times\ALP B)$. Then, we get
\beqq{& & |f_\mu(\VEC a,\VEC c)-f_\mu(\VEC a_0,\VEC c_0)| \\
&=& \left|\int_{\ALP Y\times\ALP B} g(\VEC a,\VEC b,\VEC c) \mu(d\VEC b|\VEC y) \eta(\VEC a,\VEC y) \nu(d\VEC y)-\int_{\ALP Y\times\ALP B} g(\VEC a_0,\VEC b,\VEC c_0) \mu(d\VEC b|\VEC y) \eta(\VEC a_0,\VEC y) \nu(d\VEC y)\right|,\\
&\leq &\int_{\ALP Y\times\ALP B} \Big|g(\VEC a,\VEC b,\VEC c) \eta(\VEC a,\VEC y) - g(\VEC a_0,\VEC b,\VEC c_0) \eta(\VEC a_0,\VEC y)\Big| \mu(d\VEC b|\VEC y) \nu(d\VEC y),\\
&\leq & \infnorm{g}M\epsilon+\int_{\ALP Y\times\ALP B} \Big|g(\VEC a,\VEC b,\VEC c)- g(\VEC a_0,\VEC b,\VEC c_0) \Big| \mu(d\VEC b|\VEC y) \eta(\VEC a_0,\VEC y) \nu(d\VEC y),\\
& < & (\infnorm{g}M+2)\epsilon.}
Also notice that $\{f_\mu(\cdot,\cdot)\}_{\mu\in\wp(\ALP Y\times\ALP B)}$ is bounded by $\infnorm{g}$. This completes the proof of the lemma.

\section{Proof of Lemma \ref{lem:unifequi2}}\label{app:corunifequi}
We mimic the steps of the proof of Lemma \ref{lem:unifequi} to prove this statement. First notice that $\{f_{\mu^{1:N}}(\cdot,\cdot)\}_{\mu^i\in\wp(\ALP Y^i\times\ALP B^i)}$ is uniformly bounded by $\infnorm{g}$. Now, we prove that this family of functions is equicontinuous. 

Let us define $\eta$, $M$ and $\nu$ as
\beqq{\eta(\VEC a,\VEC y^{1:N}) := \prod_{i=1}^N\eta^i(\VEC a,\VEC y^i),\quad M = \max_{i\in[N]} \sup_{\VEC a\in\ALP A}\int_{\ALP Y^i}h^id\nu^i, \quad \nu(d\VEC y^{1:N}) = \prod_{i=1}^N \nu^i(d\VEC y^i).}
First, note that for any $\VEC a\in\ALP A$ and $i\in[N]$,
\beq{\label{eqn:etaid}\int_{\ALP Y^i} \eta^i(\VEC a,\VEC y^i)\nu^i(d\VEC y^i) =\int_{\ALP Y^i} \pr{d\VEC y^i|\VEC a} = 1.}
Let $\epsilon>0$. Let $\delta^i>0$ be such that for any $d_{\ALP A}(\VEC a,\VEC a_0)<\delta^i$, we have
\beqq{|\eta^i(\VEC a,\VEC y^i)-\eta^i(\VEC a_0,\VEC y^i)|<\epsilon\: h^i(\VEC a_0,\VEC y^i).}
Pick $\delta = \min_{i\in[N]}\delta^i$. Fix $\VEC a_0\in\ALP A$. Now, for $\VEC a\in\ALP A$ such that $d_{\ALP A}(\VEC a,\VEC a_0)<\delta$, notice the following:
\beqq{\left|\eta(\VEC a,\VEC y^{1:N}) - \eta(\VEC a_0,\VEC y^{1:N})\right|  \hspace{7cm} \nonumber\\
\leq \sum_{j=1}^N \left( \prod_{i=1}^{j-1} \eta^i(\VEC a_0,\VEC y^i)\;\; \Big|\left( \eta^j(\VEC a,\VEC y^j)-\eta^j(\VEC a_0,\VEC y^j) \right)\Big| \;\prod_{i=j+1}^{N} \eta^i(\VEC a,\VEC y^i) \right)\\
\leq \epsilon \left(\sum_{j=1}^N \left( \prod_{i=1}^{j-1} \eta^i(\VEC a_0,\VEC y^i) \prod_{i=j+1}^{N} \eta^i(\VEC a,\VEC y^i) \right)h^j(\VEC a_0,\VEC y^j)\right),\hspace{1.6cm}}
where terms with $\prod_{i=1}^0$ and $\prod_{i=N+1}^{N}$ are replaced by 1. Using the above expression, we get
\beq{\left|\int_{\ALP B^{1:N}\times\ALP Y^{1:N}} g(\VEC a,\VEC b^{1:N},\VEC c)\left(\prod_{i=1}^N\mu^i(d\VEC b^i|\VEC y^i)\right) \left(\eta(\VEC a,\VEC y^{1:N}) - \eta(\VEC a_0,\VEC y^{1:N})\right)\nu(d\VEC y^{1:N})\right|\nonumber\\
< \epsilon\infnorm{g}\sum_{j=1}^N\int_{\ALP Y^{1:N}}  \left( \prod_{i=1}^{j-1}( \eta^i(\VEC a_0,\VEC y^i)d\nu^i)\prod_{i=j+1}^{N} (\eta^i(\VEC a,\VEC y^i)d\nu^i )\right) h^j(\VEC a_0,\VEC y^j)d\nu^j\nonumber\\
\leq \epsilon\infnorm{g}NM,\hspace{9cm} \label{eqn:eta1N}}
where we used \eqref{eqn:etaid}. Since $g$ is uniformly continuous, for every $\epsilon>0$, there exists a $\delta>0$ such that for all $\VEC a_0\in\ALP A$, $\VEC c_0\in\ALP C$ and $\VEC a\in\ALP A$, $\VEC c\in\ALP C$ satisfying $d_{\ALP A}(\VEC a,\VEC a_0)<\delta$, $d_{\ALP C}(\VEC c,\VEC c_0)<\delta$, we have
\beqq{\sup_{\VEC b^{1:N}\in\ALP B^{1:N}}|g(\VEC a,\VEC b^{1:N},\VEC c)-g(\VEC a_0,\VEC b^{1:N},\VEC c_0)|<2\epsilon.}
Using the inequality above and \eqref{eqn:eta1N}, one can show that
\beqq{|f_{\mu^{1:N}}(\VEC a,\VEC c)-f_{\mu^{1:N}}(\VEC a_0,\VEC c_0)|< (\infnorm{g}NM+2)\epsilon,}
which establishes the result.

\section{Proof of Lemma \ref{lem:lambdan}}\label{app:lambdan}
In order to prove the lemma, we first need the following result.
\begin{lemma}\label{lem:cnmun}
Let $\ALP A$ be a Polish space. Let $\{h_n:\ALP A\rightarrow\Re\}_{n\in\Na}$ be a convergent sequence of continuous and uniformly bounded functions and $h_0:\ALP A\rightarrow\Re$ be a continuous function such that for any compact subset $\SF A\subset\ALP A$, $\sup_{\VEC a\in\SF A}|h_n(\VEC a)-h_0(\VEC a)|\rightarrow 0$ as $n\rightarrow\infty$. Let $\{\mu_n\}_{n\in\Na\cup\{ 0\}}\subset\wp(\ALP A)$ be a weak* convergent sequence of measures such that $\mu_n\ws \mu_0$ as $n\rightarrow \infty$. Then,
\beqq{\lf{n}\int_{\ALP A} h_nd\mu_n = \int_{\ALP A} h_0d\mu_0, \quad\text{ and }\quad \lf{n}\left|\int_{\ALP A} h_nd\mu_n-\int_{\ALP A} h_nd\mu_0\right| = 0.}
\end{lemma}
\begin{proof}
Since $\{\mu_n\}_{n\in\Na}$ is a weak* convergent sequence, it is tight, which further implies that for any $\epsilon>0$, there exists a compact set $\SF A_{\epsilon}\subset\ALP A$ such that $\mu_n(\SF A_{\epsilon}^\complement)<\epsilon$ for every $n\in\Na\cup\{0\}$. Fix $\epsilon>0$. Since $h_n$ converges uniformly to $h_0$ over the compact set $\SF A_{\epsilon}$, there exists $N_{\epsilon,1}\in\Na$ such that
\beqq{|h_n(\VEC a)-h_0(\VEC a)| <\epsilon, \quad \text{ for all }\VEC a\in\SF A_{\epsilon}\text{ and } n\geq N_{\epsilon,1}.}
Let $M$ be the uniform bound on the sequence of functions $\{h_n\}_{n\in\Na}$. Then, $\infnorm{h_0}\leq M$. Thus, for any $n\geq N_{\epsilon,1}$, we get
\beq{\left|\int_{\ALP A} h_nd\mu_n-\int_{\ALP A} h_0d\mu_n\right| &\leq & \int_{\SF A_{\epsilon}} |h_n(\VEC a)-h_0(\VEC a)| d\mu_n(\VEC a) +\int_{\SF A_{\epsilon}^\complement} |h_n(\VEC a)-h_0(\VEC a)| d\mu_n(\VEC a)\nonumber\\
&<& (1+2M)\epsilon.\label{eqn:12m}}
By the definition of weak* convergence of measures, there exists $N_{\epsilon,2}\in\Na$ such that
\beqq{\left|\int_{\ALP A} h_0d\mu_n - \int_{\ALP A} h_0d\mu_0\right|<\epsilon \quad\text{ for all }n\geq N_{\epsilon,2}.}
Take $N_\epsilon = \max\{N_{\epsilon,1},N_{\epsilon,2}\}$. Now, for $n>N_{\epsilon}$, we have
\beq{\left|\int_{\ALP A} h_nd\mu_n-\int_{\ALP A} h_0d\mu_0\right| &\leq & \left|\int_{\ALP A} h_nd\mu_n-\int_{\ALP A} h_0d\mu_n\right|+\left|\int_{\ALP A} h_0 d\mu_n -\int_{\ALP A} h_0 d\mu_0\right|\nonumber\\
&<& 2(1+M)\epsilon.\label{eqn:c0m0}}
Now, consider the following inequalities for $n\geq N_{\epsilon}$:
\beqq{\left|\int_{\ALP A} h_nd\mu_n-\int_{\ALP A} h_n d\mu_0\right| &\leq&  \left|\int_{\ALP A} h_nd\mu_n-\int_{\ALP A} h_0d\mu_0\right|+\left|\int_{\ALP A} h_0d\mu_0-\int_{\ALP A} h_nd\mu_0\right|\\
& < & 2(1+M)\epsilon +(1+2M)\epsilon = (3+4M)\epsilon,}
where the first inequality is just the triangle inequality on the real line, whereas the second inequality follows from \eqref{eqn:12m} and \eqref{eqn:c0m0}. This completes the proof of the lemma.
\end{proof}

Now, we turn our attention to the proof of Lemma \ref{lem:lambdan}, which is done in two steps:

{\it Step 1:}  By Lemma \ref{lem:unifequi2}, we know that $\{f_n\}$ is a sequence of uniformly equicontinuous and uniformly bounded functions. In this step, we show that there exists a subsequence $\{f_{n_k}\}_{k\in\Na}$ and $f_0\in C_b(\ALP A\times\ALP C)$ such that $f_{n_k}$ converges to $f_0$ uniformly over any compact set in $\ALP A\times\ALP C$. 

Since $\ALP A\times\ALP C$ is $\sigma$-compact, there exists a countable collection of compact sets $\SF K_n\subset\ALP A\times\ALP C$ such that $\ALP A\times\ALP C = \cup_{n\in\Na}\SF K_n$. Let $\SF L_m = \cup_{k=1}^m \SF K_k$. By the Arzela-Ascoli Theorem \cite{ali2006}, for every $m\in\Na$, there exists a convergent subsequence $\{f_{n^m_k}\}_{k\in\Na}$ and a continuous function $f^m_0:\SF L_m\rightarrow\Re$ such that $\sup_{(\VEC a,\VEC c)\in\SF L_m} |f_{n^m_k}(\VEC a,\VEC c)-f_0^m(\VEC a,\VEC c)|\rightarrow 0$ as $k\rightarrow\infty$. We can take $\{n^{m+1}_k\}_{k\in\Na}$ to be a subsequence of $\{n^m_k\}_{k\in\Na}$ for every $m\in\Na$. Now, since $\SF L_{m}\subset\SF L_{m+1}$, we conclude that $f^m_0$ agrees with $f^{m+1}_0$ on set $\SF L_m$ for every $m\in\Na$. Using Cantor's diagonalization argument, we get a subsequence $\{f_{n_k}\}_{k\in\Na}$ and a continuous function $f_0$ such that $f_{n_k}\rightarrow f_0$, where the convergence is uniform over any compact set in $\ALP A\times\ALP C$. Furthermore,
 since $f_n$ is uniformly bounded, $f_0$ is also bounded. 

{\it Step 2:} Using the result of Lemma \ref{lem:cnmun}, we get
\beqq{\lf{k}\left|\int_{\ALP A\times\ALP C} f_{n_k}d\zeta_{n_k} - \int_{\ALP A\times\ALP C} f_{n_k}d\zeta_{0}\right| =0.}
This establishes the statement of Lemma \ref{lem:lambdan}.

\section{Proof of Lemma \ref{lem:control2}}\label{app:control2}
First, note that since pullback of a measure is a continuous operation \cite{amb2008}, $\pj^{\ALP B\times\ALP C}_{\#}\mu_n\rightarrow\pj^{\ALP B\times\ALP C}_{\#}\mu_0$. Pick any $g\in U_b(\ALP A\times\ALP B\times\ALP C)$ and note that $(\VEC b,\VEC c)\mapsto\int_{\ALP A} g(\VEC a,\VEC b,\VEC c) \rho(\VEC a,\VEC b)\nu(d\VEC a)$ is a bounded uniformly continuous function from Corollary \ref{cor:fconti}. This gives
\beqq{\int gd\mu_0 &=& \lf{n}\int gd\mu_n\\
&=& \lf{n}\int_{\ALP B\times\ALP C} \left(\int_{\ALP A} g(\VEC a,\VEC b,\VEC c) \rho(\VEC a,\VEC b)\nu(d\VEC a) \right)\pj^{\ALP B\times\ALP C}_{\#}\mu_n(d\VEC b,d\VEC c),\\
&=& \int_{\ALP B\times\ALP C} \left(\int_{\ALP A} g(\VEC a,\VEC b,\VEC c) \rho(\VEC a,\VEC b)\nu(d\VEC a) \right)\pj^{\ALP B\times\ALP C}_{\#}\mu_0(d\VEC b,d\VEC c),}
where we used disintegration of measures. This completes the proof of the lemma.

\section{Proof of Theorem \ref{thm:intc}}\label{app:intc}
Let us define a functional $\tilde J$ as:
\beqq{\tilde J(\lambda^{1:N}):=\int g(\VEC x,\VEC y^{1:N},\VEC u^{1:N}) \left(\prod_{i=1}^N \lambda^i(d\VEC u^i|\VEC y^i)\right)\pr{d\VEC x,d\VEC y^{1:N}},}
where $\lambda^i\in \wp(\ALP U^i\times\ALP Y^i\times\ALP X)$. We proceed with the proof in two steps. In the first step, we show that there exists a subsequence $\{n_{k}\}_{k\in\Na}$ such that $\lf{k}\tilde J(\lambda^{1:N}_{n_{k}}) = \tilde J(\lambda^{1:N}_0)$. Then, we show that $\lf{n}\tilde J(\lambda^{1:N}_{n}) = \tilde J(\lambda^{1:N}_0)$ using the first step.

{\it Step 1:} Consider the expressions,
\beqq{\bigg|& & \int g(\VEC x,\VEC y^{1:N},\VEC u^{1:N}) \left(\prod_{i=1}^N \lambda^i_{n}(d\VEC u^i|\VEC y^i)\right)\pr{d\VEC x,d\VEC y^{1:N}}\\
& & - \int g(\VEC x,\VEC y^{1:N},\VEC u^{1:N}) \left(\prod_{i=1}^N \lambda^i_0(d\VEC u^i|\VEC y^i)\right)\pr{d\VEC x,d\VEC y^{1:N}}\bigg|\nonumber\quad\\
&=& \bigg|\sum_{j=1}^N \int g(\VEC x,\VEC y^{1:N},\VEC u^{1:N}) \left(\prod_{i=1}^{j-1} \lambda^i_0(d\VEC u^i|\VEC y^i)\right) \left(\prod_{i=j+1}^{N} \lambda^i_{n}(d\VEC u^i|\VEC y^i)\right)\\
& & \left(\lambda^j_{n}(d\VEC u^j|\VEC y^j) -\lambda^j_0(d\VEC u^j|\VEC y^j) \right)\pr{d\VEC x,d\VEC y^{1:N}} \bigg|}
\beq{\leq  \sum_{j=1}^N \bigg|\int g(\VEC x,\VEC y^{1:N},\VEC u^{1:N}) \left(\prod_{i=1}^{j-1} \lambda^i_0(d\VEC u^i|\VEC y^i)\eta^i(\VEC x,\VEC y^i) \nu_{\ALP Y^i}(d\VEC y^i) \right) \hspace{3cm}\nonumber\\
\left(\prod_{i=j+1}^{N} \lambda^i_{n}(d\VEC u^i|\VEC y^i)\eta^i(\VEC x,\VEC y^i) \nu_{\ALP Y^i}(d\VEC y^i) \right)
\bigg(\lambda^j_{n}(d\VEC u^j,d\VEC y^j,d\VEC x) -\lambda^j_0(d\VEC u^j,d\VEC y^j, d\VEC x) \bigg) \bigg|,\quad\label{eqn:nlk}}
where the integration is taken over the space $\ALP X\times\ALP Y^{1:N}\times\ALP U^{1:N}$. Replace the product terms $\prod_{i=1}^0$ and $\prod_{i=N+1}^N$ by 1 in those expressions. In the statements of Lemma \ref{lem:unifequi} and its corollary and Lemma \ref{lem:unifequi2}, take $\ALP B = \ALP U^i\times\ALP Y^i$ for an appropriate index $i$, replace $\ALP A$ by $\ALP X$ and $\ALP C$ by appropriate product spaces.

Applying Corollary \ref{cor:fconti}, we conclude that the function
\beqq{\int_{\ALP U^{1:j-1}\times\ALP Y^{1:j-1}} g(\VEC x,\VEC y^{1:N},\VEC u^{1:N}) \left(\prod_{i=1}^{j-1} \lambda^i_0(d\VEC u^i|\VEC y^i)\eta^i(\VEC x,\VEC y^i) \nu_{\ALP Y^i}(d\VEC y^i)\right)}
is uniformly continuous in $\VEC x$, $\VEC y^{j:N}$ and $\VEC u^{j:N}$. Next, we use Lemma \ref{lem:unifequi2} to conclude that the sequence of functions
\beqq{\Bigg\{\int_{\ALP U^{-j}\times\ALP Y^{-j}} g(\VEC x,\VEC y^{1:N},\VEC u^{1:N})\prod_{i=1}^{j-1} \lambda^i_0(d\VEC u^i|\VEC y^i)\eta^i(\VEC x,\VEC y^i) \nu_{\ALP Y^i}(d\VEC y^i)\\
\left(\prod_{i=j+1}^{N} \lambda^i_{n}(d\VEC u^i|\VEC y^i)\eta^i(\VEC x,\VEC y^i) \nu_{\ALP Y^i}(d\VEC y^i)\right)\Bigg\}_{n\in\Na}}
is uniformly equicontinuous and bounded on $\ALP X\times\ALP Y^j\times\ALP U^j$ for every $j\in[N]$. Then, there exists a subsequence $\{n_{l}\}_{l\in\Na}$ by Lemma \ref{lem:lambdan} for $j=1$ such that as $l\rightarrow\infty$,
\beqq{ \bigg|\int \left(g(\VEC x,\VEC y^{1:N},\VEC u^{1:N}) \left(\prod_{i=2}^{N} \lambda^i_{n_{l}}(d\VEC u^i|\VEC y^i)\eta^i(\VEC x,\VEC y^i) \nu_{\ALP Y^i}(d\VEC y^i)\right)\right)\\
 \left(\lambda^1_{n_{l}}(d\VEC u^1,d\VEC y^1,d\VEC x) -\lambda^1_0(d\VEC u^1,d\VEC y^i,d\VEC x) \right) \bigg|\rightarrow 0.}
Along the sequence $\{n_{l}\}_{l\in\Na}$, there exists a further subsequence $\{n_{l_m}\}_{m\in\Na}$ for $j=2$ such that second term in the summation  in \eqref{eqn:nlk} goes to zero as $k\rightarrow\infty$. Continue this process for $j=3,\ldots,N$ to arrive at a subsequence $\{n_{k}\}_{k\in\Na}$ such that each component of the sum in \eqref{eqn:nlk} converges to $0$ as $k\rightarrow\infty$. Thus, we get
\beqq{\lf{k}\tilde J(\lambda^{1:N}_{n_{k}}) = \tilde J(\lambda^{1:N}_0).}

{\it Step 2:} We now claim that $\lf{n}\tilde J(\lambda^{1:N}_{n}) = \tilde J(\lambda^{1:N}_0)$, which we prove by contradiction. Suppose that $\lf{n}\tilde J(\lambda^{1:N}_{n})$ does not exist or is not equal to $\tilde J(\lambda^{1:N}_0)$. In this case, there exists an $\epsilon_0>0$ and a subsequence $\{n_m\}_{m\in\Na}$ such that
\beqq{|\tilde J(\lambda^{1:N}_{n_m}) - \tilde J(\lambda^{1:N}_0)|>\epsilon_0 \quad \text{ for all }m\in\Na.}
From Step 1 of the proof, we know that there exists a further subsequence $\{n_{m_k}\}_{k\in\Na}$ such that
\beqq{\lf{k}\tilde J(\lambda^{1:N}_{n_{m_k}}) = \tilde J(\lambda^{1:N}_0),}
which is a contradiction. Thus, $\lf{n}\tilde J(\lambda^{1:N}_{n}) = \tilde J(\lambda^{1:N}_0)$, which completes the proof of the first part of the theorem.

Since the first part of the lemma holds for all uniformly continuous functions, we arrive at the second result by \cite[Theorem 9.1.5, p. 372]{stroock2011}.

\section{Proof of Lemma \ref{lem:intcdegraded}}\label{app:intcdegraded}
First, note that the information constraints of the limits $\lambda^1_0$ and $\lambda^2_0$ are satisfied due to Assumption \ref{assm:degraded} and Lemma \ref{lem:control2}. To establish the result, we follow the same steps as in the proof of Theorem \ref{thm:intc} in Appendix \ref{app:intc} above with some minor modifications. Consider the following expressions:
\beq{& & \bigg|\int g(\VEC x,\VEC y^{1:2},\VEC u^{1:2}) \left(\prod_{i=1}^2 \lambda^i_{n}(d\VEC u^i|\VEC y^i)\right)\pr{d\VEC x,d\VEC y^{1:2}}\nonumber\\
& & - \int g(\VEC x,\VEC y^{1:2},\VEC u^{1:2}) \left(\prod_{i=1}^2 \lambda^i_0(d\VEC u^i|\VEC y^i)\right)\pr{d\VEC x,d\VEC y^{1:2}}\bigg|\label{eqn:cdeg}\\
&\leq & \bigg|\int g(\VEC x,\VEC y^{1:2},\VEC u^{1:2}) \lambda^2_n(d\VEC u^2|\VEC y^2) \left(\lambda^1_{n}(d\VEC u^1|\VEC y^1) -\lambda^1_0(d\VEC u^1|\VEC y^1) \right)\pr{d\VEC x,d\VEC y^{1:2}} \bigg|\nonumber\\
& & +\bigg|\int g(\VEC x,\VEC y^{1:2},\VEC u^{1:2}) \lambda^1_0(d\VEC u^1|\VEC y^1) \left(\lambda^2_{n}(d\VEC u^2|\VEC y^2) -\lambda^2_0(d\VEC u^2|\VEC y^2) \right)\pr{d\VEC x,d\VEC y^{1:2}} \bigg|\nonumber,}
where the integration is taken over the space $\ALP X\times\ALP Y^{1:2}\times\ALP U^{1:2}$. We use Lemma \ref{lem:unifequi} to conclude that the sequence of functions
\beqq{\Bigg\{\int_{\ALP U^2\times\ALP Y^2} g(\VEC x,\VEC y^{1:2},\VEC u^{1:2}) \lambda^2_n(d\VEC u^2|\VEC y^2) \eta^2(\VEC y^1,\VEC y^2) \nu_{\ALP Y^2}(d\VEC y^2)\Bigg\}_{n\in\Na}}
is uniformly equicontinuous and bounded on $\ALP X\times\ALP Y^1\times\ALP U^1$. Then, there exists a subsequence $\{n_l\}_{l\in\Na}$ by Lemma \ref{lem:lambdan} such that as $l\rightarrow\infty$,
\beqq{ \bigg|\int g(\VEC x,\VEC y^{1:2},\VEC u^{1:2}) \lambda^2_{n_l}(d\VEC u^2|\VEC y^2) \left(\lambda^1_{n_l}(d\VEC u^1|\VEC y^1) -\lambda^1_0(d\VEC u^1|\VEC y^1) \right)\pr{d\VEC x,d\VEC y^{1:2}} \bigg|\rightarrow 0.}
We now apply Corollary \ref{cor:fconti} to conclude that the function
\beqq{\int_{\ALP U^1\times\ALP Y^1} g(\VEC x,\VEC y^{1:2},\VEC u^{1:2}) \lambda^1_0(d\VEC u^1|\VEC y^1)\eta^1(\VEC x,\VEC y^2,\VEC y^1) \nu_{\ALP Y^1}(d\VEC y^1)}
is uniformly continuous in $\VEC x$, $\VEC y^{2}$ and $\VEC u^{2}$. Thus, as $l\rightarrow\infty$, we get
\beqq{\bigg|\int g(\VEC x,\VEC y^{1:2},\VEC u^{1:2}) \lambda^1_0(d\VEC u^1|\VEC y^1) \left(\lambda^2_{n_l}(d\VEC u^2|\VEC y^2) -\lambda^2_0(d\VEC u^2|\VEC y^2) \right)\pr{d\VEC x,d\VEC y^{1:2}} \bigg|\rightarrow 0.}
This implies that along the subsequence $\{n_l\}_{l\in\Na}$, \eqref{eqn:cdeg} converges to $0$ as $l\rightarrow\infty$. We can now mimic Step 2 of the proof of Theorem \ref{thm:intc} in Appendix \ref{app:intc} to complete the proof of the lemma.

\section{Proof of Theorem \ref{thm:sub2}}\label{app:sub2}
By Theorem \ref{thm:sub},
\beqq{\lf{n} \int(\min\{c,m\}) \prod_{i=1}^N\lambda^i_{n}(d\VEC u^i|\VEC y^i)\pr{d\VEC x,d\VEC y^{1:N}}\qquad\qquad\qquad\\
= \int(\min\{c,m\})\; \prod_{i=1}^N\lambda^i_0(d\VEC u^i|\VEC y^i)\pr{d\VEC x,d\VEC y^{1:N}} \text{ for all } m\in\Na.}
Notice that $\min\{c,m\}\nearrow c$ as $m\rightarrow\infty$. We get
\beqq{& & \lif{n}\int c\;\prod_{i=1}^N\lambda^i_{n}(d\VEC u^i|\VEC y^i)\pr{d\VEC x,d\VEC y^{1:N}} \\
&\geq & \lif{k}\int \min\{c,m\}\;\prod_{i=1}^N\lambda^i_{n}(d\VEC u^i|\VEC y^i)\pr{d\VEC x,d\VEC y^{1:N}} \\
 & = & \int(\min\{c,m\})\; \prod_{i=1}^N\lambda^i_0(d\VEC u^i|\VEC y^i)\pr{d\VEC x,d\VEC y^{1:N}}.}
The left-side of the equation is independent of $m$ and the right-side of the equation holds for any $m\in\Na$. Taking the limit as $m\rightarrow\infty$ and using the monotone convergence theorem, we get
\beqq{\lif{n}\int c\;\prod_{i=1}^N\lambda^i_{n}(d\VEC u^i|\VEC y^i)\pr{d\VEC x,d\VEC y^{1:N}} \geq\int c\;\prod_{i=1}^N\lambda^i_0(d\VEC u^i|\VEC y^i)\pr{d\VEC x,d\VEC y^{1:N}}.}
This establishes the theorem.

\section{Proof of Lemma \ref{lem:q2kcom}}\label{app:q2kcom}
First, we recall a general version of Markov's inequality.
\begin{lemma}[Generalized Markov's Inequality]
\label{lem:cheb}
Let $\phi:\ALP A\rightarrow \Re$ be a non-negative measurable function and $\SF A\subset\ALP A$ be a Borel measurable set such that $\inf_{\VEC a\in \SF A} \phi(\VEC a)>0$. Then,
\beqq{\pr{\SF A}\leq \frac{\ex{\phi\IND_{\SF A}}}{\inf_{\VEC a\in \SF A} \phi(\VEC a)}.}
\end{lemma}
\begin{proof}
Note that $\left(\inf_{\VEC a\in \SF A} \phi(\VEC a)  \right)\IND_{\SF A}(\VEC a) \leq \phi(\VEC a)\IND_{\SF A}(\VEC a)$ for all $\VEC a\in \SF A$. Taking expectations on both sides leads us to the result.
\end{proof}

We want to show that the set of measures in $\ALP M$ is tight. Toward this end, we fix $\epsilon>0$, and show that there exist compact sets $\SF K_\epsilon\subset \ALP A$ and $\SF L_\epsilon\subset \ALP B$ such that $\pj^{\ALP A\times\ALP B}_{\#}\mu(\SF K_\epsilon\times \SF L_\epsilon) =\mu(\SF K_\epsilon\times \SF L_\epsilon\times\ALP C) >1-2\epsilon$ for all $\mu\in\ALP M$. This proves that the set of measures $\ALP M$ is tight.

Since the set of measures $\ALP N$ is tight, there exists a compact set $\SF K_{\epsilon}\subset\ALP A$ such that $\zeta(\SF K_{\epsilon}^\complement)<\epsilon$ for all $\zeta\in\ALP N$. Pick $M\in\Re^+$ sufficiently large such that $M>k/\epsilon$. We carry out the analysis for the two cases separately.
\begin{enumerate}
\item  Assume that $\phi$ satisfies the first condition in Definition \ref{def:icclass}. Given $M$ and $\SF K_{\epsilon}$, let $\SF L_\epsilon\subset \ALP B$ be the compact set such that
\beqq{\inf_{(\VEC a,\VEC b,\VEC c)\in (\SF K_{\epsilon}\times \SF L_{\epsilon}^\complement\times\ALP C)} \phi(\VEC a,\VEC b,\VEC c)\geq M.}
Now note that $\mu(\SF K_\epsilon^\complement\times \SF B\times\ALP C)<\epsilon$ for all Borel sets $\SF B\subset\ALP B$. Let $\SF E =(\SF K_{\epsilon}\times \SF L_{\epsilon})^\complement \times\ALP C$ and note that
$\SF E=  (\SF K_\epsilon\times \SF L_\epsilon^\complement\times\ALP C)\bigcup (\SF K_\epsilon^\complement\times \ALP B\times\ALP C)$.


Define $\SF E_1 = \SF K_{\epsilon}\times \SF L_{\epsilon}^\complement\times\ALP C$ and $\SF E_2 =  (\SF K_\epsilon^\complement\times \ALP B\times\ALP C)$. It is easy to verify that $\SF E_1\bigcap \SF E_2=\emptyset$. We now use generalized Markov's inequality (Lemma \ref{lem:cheb}) to get
\beqq{\mu(\SF K_{\epsilon}\times \SF L_{\epsilon}^\complement\times\ALP C) =\mu(\SF E_1)\leq \frac{\ex{\phi 1_{\SF E_1}}}{\inf_{x\in \SF E_1} \phi(x)}\leq \frac{k}{M}<\epsilon}
for all $\mu\in\ALP M$, where $1_{\SF E_1}$ is the indicator function over the set $\SF E_1$. By the additivity property of probability measures, we get $\mu(\SF E)=\mu(\SF E_1)+\mu(\SF E_2)<2\epsilon$.


\item Now suppose that $\phi$ satisfies the second assumption. For every $a\in \SF K_{\epsilon}$, let $\SF O_a\subset \ALP A$ be the open neighborhood of $a\in \SF K_{\epsilon}$ and $\SF L_{\VEC a}\subset\ALP B$ be the compact set in $\ALP B$ such that
\beqq{\inf_{(\VEC a,\VEC b,\VEC c)\in (\SF O_a\times \SF L_a^\complement\times\ALP C)} \phi(\VEC a,\VEC b,\VEC c)\geq M.}
Notice that $\{\SF O_{\VEC a}\}_{\VEC a\in \SF K_{\epsilon}}$ is an open cover for $\SF K_\epsilon$. By the definition of compactness, there exists a finite subcover, say $\{\SF O_{\VEC a_j}\}_{j=1}^n$, such that $\SF K_\epsilon\subset \SF O_\epsilon :=\bigcup_{j=1}^n \SF O_{\VEC a_j}$.
Define $\SF L_\epsilon := \bigcup_{j=1}^n \SF L_{\VEC a_j}\subset\ALP B$, which is a compact set in $\ALP B$. Since $\SF L_{\VEC a_j}\subset \SF L_{\epsilon}$, we get $\SF O_{\VEC a_j}\times \SF L_{\epsilon}^\complement\subset \SF O_{\VEC a_j}\times \SF L_{\VEC a_j}^\complement$ and as a result of this inclusion, we conclude
\beqq{\inf_{(\VEC a,\VEC b,\VEC c)\in (\SF O_{\VEC a_j}\times \SF L_{\epsilon}^\complement\times\ALP C)} \phi(\VEC a,\VEC b,\VEC c)\geq \inf_{(\VEC a,\VEC b,\VEC c)\in (\SF O_{\VEC a_j}\times \SF L_{\VEC a_j}^\complement\times\ALP C)} \phi(\VEC a,\VEC b,\VEC c)\geq M,}
for all $j\in\{1,\ldots,n\}$. Notice that $\bigcup_{j=1}^n (\SF O_{\VEC a_j}\times \SF L_{\epsilon}^\complement) =  \SF O_\epsilon\times \SF L_{\epsilon}^\complement$.
Since over each set $\SF O_{\VEC a_j}\times \SF L_{\epsilon}^\complement$, the infimum of $\phi$ is greater than or equal to $M$, we conclude that
\beqq{\inf_{(\VEC a,\VEC b,\VEC c)\in (\SF O_{\epsilon}\times \SF L_{\epsilon}^\complement\times\ALP C)} \phi(\VEC a,\VEC b,\VEC c)\geq M.}
Now again define $\SF E :=(\SF K_{\epsilon}\times \SF L_{\epsilon})^\complement\times\ALP C $, $\SF E_1 := \SF O_{\epsilon}\times \SF L_{\epsilon}^\complement \times\ALP C$, $\SF E_2 :=  (\SF K_\epsilon^\complement\times \ALP B\times\ALP C)$, and note that $\SF E\subsetneq \SF E_1\bigcup \SF E_2 $. By a similar argument as in Part 1 of this proof, we conclude that $\mu(\SF E_1)<\epsilon$ and $\mu(\SF E_2)<\epsilon$, which means $\mu(\SF E)<2\epsilon$.
\end{enumerate}
Note that $\SF K_\epsilon\times \SF L_\epsilon$ is a compact set in $\ALP A\times\ALP B$ and its complement has small measure. Thus, we conclude that the set of probability measures $\pj^{\ALP A\times\ALP B}_{\#}\ALP M$ is tight.

If $\phi$ is lower semi-continuous, then $\mu \mapsto \int \phi\;d\mu$ is a lower semi-continuous functional \cite[Lemma 4.3]{villani2009}. Thus, we conclude that $\ALP M$ is in fact weak* closed, and therefore $\ALP M$ is weak* compact. This completes the proof of the lemma.

\section{Proof of Lemma \ref{lem:mittight}}\label{app:mittight}
Assumption \ref{ass:one} and the structure of the cost function of the team as defined in \eqref{eqn:dyncost} imply that for any $i\in[N]$ and $t\in[T]$, we have
\beqq{\int_{\Omega_0\times\ALP Y^{1:N}_{1:t-1}\times\ALP U^{1:N}_{1:t-1}\times\ALP Y^i_t\times\ALP U^i_t} \bar{c}^i_t(\VEC u^i_t,\omega_0,\VEC u^{1:N}_{1:t-1},\VEC y^{1:N}_{1:t-1},\VEC y^i_t)\;d\lambda^{1:N}_{1:t-1}\;d\lambda^i_t\; \pr{d\omega_0}\leq J(\tilde\pi^{1:N}_{1:T})}
for any choice of $\lambda^{1:N}_{1:T}\in\ALP N^{1:N}_{1:T}$. Also recall from Lemma \ref{lem:cbarit} that $\bar{c}^i_t$ is in class $\ic{\Omega_0\times\ALP U^{1:N}_{1:t-1}\times\ALP Y^{1:N}_{1:t-1}\times\ALP Y^i_t,\ALP U^i_t}$. We now use the result of Lemma \ref{lem:q2kcom} and the principle of mathematical induction to prove the result.

{\it Step 1:} Consider Agent $(1,1)$. We know that $\bar{c}^1_1$ is a non-negative continuous function in class $\ic{\Omega_0\times\ALP Y^1_1,\ALP U^1_1}$. Moreover, the measure on $\Omega_0\times\ALP Y^1_1$ is tight. Using the result of Lemma \ref{lem:q2kcom}, we get that $\ALP N^1_1$ is a tight set of measures.

{\it Step 2:} Using the same argument as in Step 1, we conclude that $\ALP M^i_1$ is a tight set of measures for all $i\in[N]$.

{\it Step 3:} Let us assume that $\ALP M^i_s$ is a tight set of measures for all $i\in[N]$ and $1\leq s\leq t-1$. Consider any Agent $(i,t)$. We know that $\bar{c}^i_t$ is a non-negative continuous function in class $\ic{\Omega_0\times\ALP U^{1:N}_{1:t-1}\times\ALP Y^{1:N}_{1:t-1}\times\ALP Y^i_t,\ALP U^i_t}$. Moreover, the set of all possible measures on $\Omega_0\times\ALP U^{1:N}_{1:t-1}\times\ALP Y^{1:N}_{1:t-1}$ induced by $\ALP N^{1:N}_{1:t-1}$ is tight because $\ALP M^i_s$ is tight for all $i\in[N]$ and $1\leq s\leq t-1$ by the induction hypothesis. Therefore, using the result of Lemma \ref{lem:q2kcom}, we get that $\ALP M^i_t$ is a tight set of measures.

This completes the induction step and we conclude that $\ALP M^i_t$ is tight for all $i\in[N]$ and $t\in[T]$. This completes the proof of the lemma.

\bibliographystyle{ieeetr}
\bibliography{nonclass,mypaper,math,game,markov}

\begin{thebibliography}{10}

\bibitem{gupta2014teamacc}
A.~Gupta, S.~Y\"uksel, C.~Langbort, and T.~Ba\c{s}ar, ``On the existence of
  optimal strategies in multi-agent stochastic teams,'' in {\em Proc. 2014
  American Control Conference (ACC)}, June 2014.

\bibitem{gupta2014teamcdc}
A.~Gupta, S.~Y\"uksel, and T.~Ba\c{s}ar, ``On the existence of optimal
  strategies in a class of dynamic stochastic teams,'' {\em submitted to 53rd
  IEEE Conference in Decision and Control (CDC)}, March 2014.

\bibitem{teneket1996}
D.~Teneketzis, ``On information structures and nonsequential stochastic
  control,'' {\em Centrum voor Wiskunde en Informatica Quarterly}, vol.~9,
  no.~4, pp.~241--261, 1996.

\bibitem{wit1971b}
H.~Witsenhausen, ``On information structures, feedback and causality,'' {\em
  SIAM Journal on Control}, vol.~9, no.~2, pp.~149--160, 1971.

\bibitem{yukselbook}
S.~Y\"uksel and T.~Ba\c{s}ar, {\em Stochastic Networked Control Systems:
  Stabilization and Optimization under Information Constraints}.
\newblock Boston, MA: Birkh\"auser, 2013.

\bibitem{radner1962}
R.~Radner, ``Team decision problems,'' {\em The Annals of Mathematical
  Statistics}, vol.~33, no.~3, pp.~857--881, 1962.

\bibitem{krainak1982}
J.~Krainak, J.~Speyer, and S.~Marcus, ``Static team problems--part {I}:
  Sufficient conditions and the exponential cost criterion,'' {\em IEEE
  Transactions on Automatic Control}, vol.~27, no.~4, pp.~839--848, 1982.

\bibitem{ho1972a}
Y.~Ho and K.~Chu, ``Team decision theory and information structures in optimal
  control problems--{P}art {I},'' {\em IEEE Transactions on Automatic Control},
  vol.~17, no.~1, pp.~15--22, 1972.

\bibitem{chu1972}
K.~C. Chu, ``Team decision theory and information structures in optimal control
  problems--{P}art {II},'' {\em IEEE Transactions on Automatic Control},
  vol.~17, pp.~22--28, February 1972.

\bibitem{witcount}
H.~Witsenhausen, ``{A counterexample in stochastic optimum control},'' {\em
  SIAM Journal on Control}, vol.~6, pp.~131--147, 1968.

\bibitem{NayyarBookChapter}
A.~Nayyar, A.~Mahajan, and D.~Teneketzis, ``The common-information approach to
  decentralized stochastic control,'' in {\em {\em Information and Control in
  Networks}, Editors: G. Como, B. Bernhardsson, A. Rantzer}, Springer, 2013.

\bibitem{CDCTutorial}
A.~Mahajan, N.~Martins, M.~Rotkowitz, and S.~Y\"uksel, ``Information structures
  in optimal decentralized control,'' in {\em Proc. of 51st IEEE Conf. Decision
  and Control}, (Hawaii, USA), 2012.

\bibitem{wit1971a}
H.~Witsenhausen, ``Separation of estimation and control for discrete time
  systems,'' {\em Proceedings of the IEEE}, vol.~59, no.~11, pp.~1557--1566,
  1971.

\bibitem{wit1988}
H.~Witsenhausen, ``Equivalent stochastic control problems,'' {\em Mathematics
  of Control, Signals, and Systems (MCSS)}, vol.~1, no.~1, pp.~3--11, 1988.

\bibitem{bambos2013b}
C.~D. Charalambous and N.~U. Ahmed, ``Dynamic team theory of stochastic
  differential decision systems with decentralized noisy information structures
  via {G}irsanov's measure transformation.''
  \url{http://arxiv.org/abs/1309.1913}, 2013.
\newblock [Online; accessed 25-Sep-2013].

\bibitem{wu2011}
Y.~Wu and S.~Verd{\'u}, ``Witsenhausen's counterexample: A view from optimal
  transport theory,'' in {\em Proc. 50th IEEE Conf. Decision and Control and
  European Control Conference (CDC-ECC)}, pp.~5732--5737, 2011.

\bibitem{basar2008}
T.~Ba\c{s}ar, ``Variations on the theme of the {W}itsenhausen counterexample,''
  in {\em Proc. 47th IEEE Conf. Decision and Control}, pp.~1614--1619, 2008.

\bibitem{cover2006}
T.~Cover and J.~Thomas, {\em Elements of Information Theory}.
\newblock Wiley-Interscience, 2006.

\bibitem{bansal1987}
R.~Bansal and T.~Ba\c{s}ar, ``Stochastic teams with nonclassical information
  revisited: When is an affine law optimal?,'' {\em IEEE Transactions on
  Automatic Control}, vol.~32, no.~6, pp.~554--559, 1987.

\bibitem{puterman1994}
M.~Puterman, {\em Markov Decision Processes: Discrete Stochastic Dynamic
  Programming}.
\newblock John Wiley \& Sons, Inc., 1994.

\bibitem{her1996}
O.~Hern{\'a}ndez-Lerma and J.~Lasserre, {\em Discrete-time Markov Control
  Processes: Basic Optimality Criteria}.
\newblock Springer, New York, 1996.

\bibitem{bert1978}
D.~P. Bertsekas and S.~E. Shreve, {\em Stochastic Optimal Control: The Discrete
  Time Case}, vol.~139.
\newblock Academic Press, New York, 1978.

\bibitem{amb2008}
L.~Ambrosio, N.~Gigli, G.~Savar{\'e}, and M.~Struwe, ``Gradient flows: in
  metric spaces and in the space of probability measures,'' {\em Lectures in
  Mathematics ETH Z{\"u}rich}, 2008.

\bibitem{villani2009}
C.~Villani, {\em Optimal Transport: Old and New}, vol.~338.
\newblock Springer, 2009.

\bibitem{yuksel2012opt}
S.~Y{\"u}ksel and T.~Linder, ``Optimization and convergence of observation
  channels in stochastic control,'' {\em SIAM Journal on Control and
  Optimization}, vol.~50, no.~2, pp.~864--887, 2012.

\bibitem{borkar1993}
V.~S. Borkar, ``White-noise representations in stochastic realization theory,''
  {\em SIAM Journal on Control and Optimization}, vol.~31, no.~5,
  pp.~1093--1102, 1993.

\bibitem{borkar2005}
V.~Borkar, S.~Mitter, A.~Sahai, and S.~Tatikonda, ``Sequential source coding:
  an optimization viewpoint,'' in {\em Proc. of 44th IEEE Conf. Decision and
  Control and European Control Conference (CDC-ECC)}, pp.~1035--1042, 2005.

\bibitem{ali2006}
C.~Aliprantis and K.~Border, {\em Infinite Dimensional Analysis: A Hitchhiker's
  Guide}.
\newblock Springer, 2006.

\bibitem{durrett2010}
R.~Durrett, {\em Probability: Theory and Examples}.
\newblock Cambridge Univ Press, 2010.

\bibitem{blackwell1963}
D.~Blackwell and C.~Ryll-Nardzewski, ``Non-existence of everywhere proper
  conditional distributions,'' {\em The Annals of Mathematical Statistics},
  vol.~34, no.~1, pp.~223--225, 1963.

\bibitem{blackwell1964}
D.~Blackwell, ``Memoryless strategies in finite-stage dynamic programming,''
  {\em The Annals of Mathematical Statistics}, vol.~35, no.~2, pp.~863--865,
  1964.

\bibitem{bogachev2006b}
V.~Bogachev, {\em Measure Theory}, vol.~2.
\newblock Springer, 2006.

\bibitem{basar1976}
T.~Ba\c{s}ar, ``On the optimality of nonlinear designs in the control of linear
  decentralized systems,'' {\em IEEE Transactions on Automatic Control},
  vol.~21, p.~797, Oct 1976.

\bibitem{gastpar2008uncoded}
M.~Gastpar, ``Uncoded transmission is exactly optimal for a simple {G}aussian
  ``sensor'' network,'' {\em IEEE Transactions on Information Theory}, vol.~54,
  no.~11, pp.~5247--5251, 2008.

\bibitem{grover2010}
P.~Grover and A.~Sahai, ``Witsenhausen's counterexample as assisted
  interference suppression,'' {\em International Journal of Systems, Control
  and Communications}, vol.~2, no.~1, pp.~197--237, 2010.

\bibitem{lipsa2011}
G.~Lipsa and N.~Martins, ``Optimal memoryless control in {G}aussian noise: A
  simple counterexample,'' {\em Automatica}, vol.~47, no.~3, pp.~552--558,
  2011.

\bibitem{zaidi2011}
A.~A. Zaidi, S.~Y\"{u}ksel, T.~J. Oechtering, and M.~Skoglund, ``On optimal
  policies for control and estimation over a {G}aussian relay channel,'' in
  {\em Proc. 50th IEEE Conf. Decision and Control and European Control
  Conference (CDC-ECC)}, pp.~5720--5725, Dec. 2011.

\bibitem{stroock2011}
D.~W. Stroock, {\em Probability Theory: An Analytic View}.
\newblock Cambridge University Press, 2011.

\end{thebibliography}

\end{document}